\newif\ifarxiv
\arxivtrue
\ifdefined\pdfminorversion
\fi
\ifdefined\pdfobjcompresslevel
\fi
\documentclass[conference]{IEEEtran}
\IEEEoverridecommandlockouts
\usepackage{cite}
\usepackage{amsmath,amssymb,amsfonts}
\usepackage{amsthm}
\usepackage{graphicx}
\usepackage{textcomp}
\usepackage{xcolor}
\usepackage[hyphens]{url}
\usepackage{array}
\usepackage{booktabs}
\usepackage{algorithm}
\usepackage[noend]{algpseudocode}
\usepackage{braket}
\usepackage{comment}
\usepackage{stmaryrd}
\usepackage{enumitem}

\newcommand{\rvs}[2]{#2}
\providecommand{\Description}[1]{}

\makeatletter
\@ifundefined{ifarxiv}{\newif\ifarxiv\arxivfalse}{}
\makeatother
\ifarxiv
\usepackage{eso-pic}
\fi

\newcommand{\secref}[1]{Sec.~\ref{#1}}
\newcommand{\figref}[1]{Fig.~\ref{#1}}
\newcommand{\tabref}[1]{Table~\ref{#1}}
\newcommand{\alref}[1]{Algorithm~\ref{#1}}

\definecolor{ReviewerA}{RGB}{90,0,130}
\definecolor{Shepherd}{RGB}{30,110,220}

\newtheorem{theorem}{Theorem}
\newtheorem{lemma}[theorem]{Lemma}
\newtheorem{premise}{Premise}

\newtheorem{definition}{Definition}

\DeclareMathOperator*{\argmax}{argmax}
\def\BibTeX{{\rm B\kern-.05em{\sc i\kern-.025em b}\kern-.08em
    T\kern-.1667em\lower.7ex\hbox{E}\kern-.125emX}}
\ifarxiv
\usepackage[hidelinks]{hyperref}
\fi

\begin{document}

\title{Coset Ensemble Decoder for Quantum Error Correction with Algorithm--Hardware Co-Design}

\author{
    \ifarxiv\thanks{Accepted to appear in the 53rd Annual International Symposium on Computer Architecture (ISCA 2026).}\fi
    \IEEEauthorblockN{
        Shuang Liang\IEEEauthorrefmark{1},
        Jubo Xu\IEEEauthorrefmark{1},
        Giulio Bassanino\IEEEauthorrefmark{1},
        Qianzhou Wang\IEEEauthorrefmark{1}, 
        Yidong Zhou\IEEEauthorrefmark{2},
        Yuncheng Lu\IEEEauthorrefmark{1},
        Zhiwen Mo\IEEEauthorrefmark{1}, \\
        Paul H. J. Kelly\IEEEauthorrefmark{1},
        Bo Yuan\IEEEauthorrefmark{2},
        Wayne Luk\IEEEauthorrefmark{1}, and
        Hongxiang Fan\IEEEauthorrefmark{1}
    }
    \vspace{0.15cm}
    \IEEEauthorblockA{\IEEEauthorrefmark{1}Imperial College London, United Kingdom \IEEEauthorrefmark{2}Rutgers University, United States}
    \IEEEauthorblockA{ \{shuang.liang, jubo.xu20, p.kelly, w.luk, hongxiang.fan\}@imperial.ac.uk}
}

\ifarxiv
\newcommand{\arxivartifactbadges}{%
  \AddToShipoutPictureFG*{%
    \AtPageUpperLeft{%
      \put(\LenToUnit{\paperwidth},\LenToUnit{-0.625in}){%
        \makebox[0pt][r]{%
          \href{https://www.acm.org/publications/policies/artifact-review-and-badging-current}{\includegraphics[width=1.35cm]{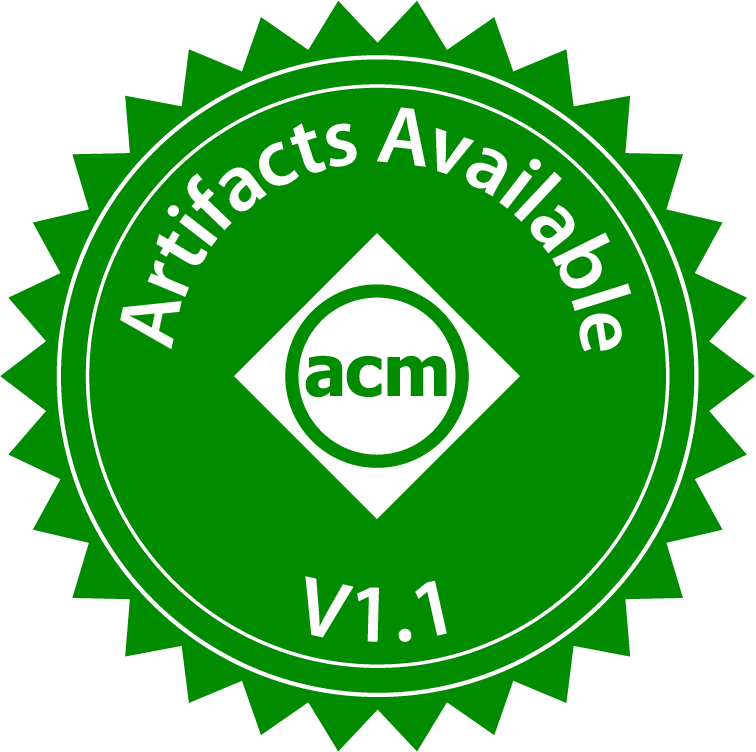}}\hspace{0.18cm}%
          \href{https://www.acm.org/publications/policies/artifact-review-and-badging-current}{\includegraphics[width=1.35cm]{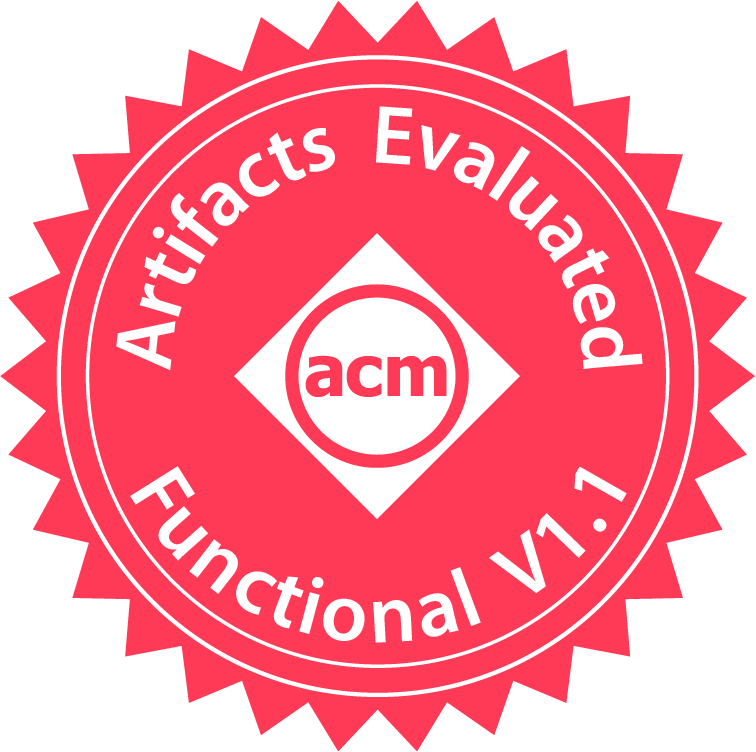}}\hspace{0.18cm}%
          \href{https://www.acm.org/publications/policies/artifact-review-and-badging-current}{\includegraphics[width=1.35cm]{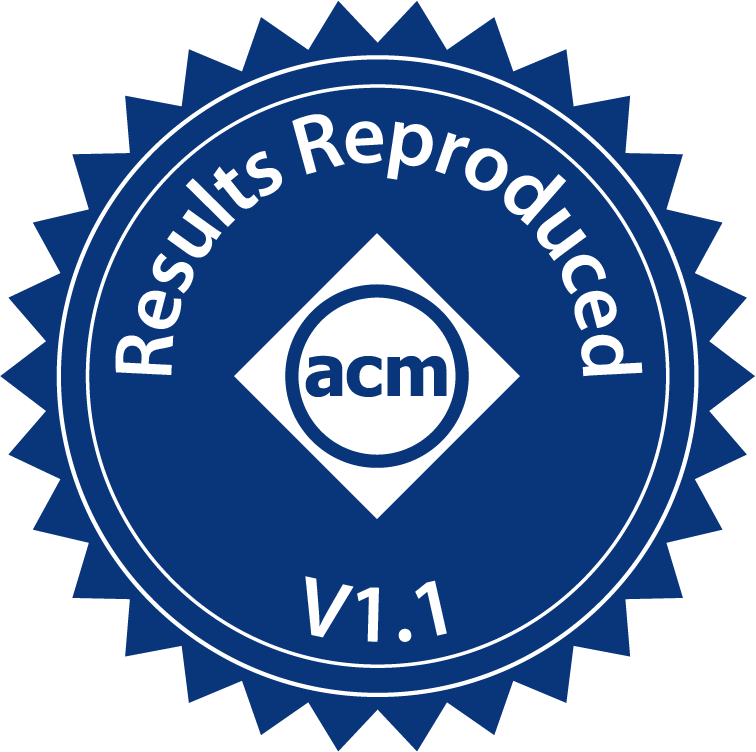}}%
          \hspace{0.30in}%
        }%
      }%
    }%
  }%
}
\fi

\ifarxiv
\arxivartifactbadges
\fi
\maketitle

\begin{abstract}
Reliable large-scale quantum computation relies on fault-tolerant architectures, where quantum error correction (QEC) continuously extracts and decodes error syndromes in real time.
A critical component in QEC is the decoder, a classical subsystem that must simultaneously deliver high logical accuracy and ultra-low latency.
This paper presents a novel algorithm--hardware co-design that improves the accuracy--latency trade-off over existing approaches such as vanilla Minimum-Weight Perfect Matching (MWPM) and Union-Find (UF) decoders.
At the algorithmic level, we introduce {coset ensemble decoding}, which improves UF decoding by explicitly exploiting logically equivalent cosets. 
Our method performs ensemble forest exploration to generate multiple coset-consistent candidates and aggregates them to approximate coset-level maximum-likelihood decoding. 
We further reduce computational and memory complexity via reverse-order elimination and lossless graph compression, without sacrificing accuracy. 
At the hardware level, we design a domain-specific architecture that temporally reuses resources, avoiding the code-distance-proportional resource growth in prior spatial architectures.
Several optimizations, such as multi-bank memory hashing and hierarchical ID mapping, are proposed to mitigate pipeline stalls and memory conflicts under highly concurrent access patterns.
Under a circuit-level depolarizing noise model, our co-design approach achieves a better accuracy--latency trade-off than prior MWPM- and UF-based decoders, while reducing FPGA LUT consumption by up to 8.2 times compared with reported UF-based decoder resources.
The tunable candidate number further exposes a flexible design knob, enabling users to tailor decoding performance to the requirements of different fault-tolerant workloads.
Our implementation is publicly available at \url{https://github.com/IMSeonL/coset-ensemble-decoder}.

\end{abstract}


\section{Introduction}

Recent advances in quantum computing have led to rapid growth in the number of physical qubits.
However, existing quantum devices continue to suffer from gate errors, crosstalk, leakage, and limited coherence times~\cite{kjaergaard2020superconducting}, all of which hinder the execution of practical quantum circuits at scale. 
Achieving practical, large-scale quantum computation therefore requires fault-tolerant quantum computing (FTQC), which relies on quantum error correction (QEC) to encode logical qubits across multiple physical qubits.
QEC continuously extracts error syndromes and corrects errors during runtime~\cite{terhal2015quantum}, suppressing the logical error rate below a desired threshold as guaranteed by the threshold theorem~\cite{aharonov1997ftthreshold, knill1998resilient}.

At the core of QEC is the decoder, a classical subsystem that interprets error syndromes and generates corrections in real time. Designing such decoders is challenging due to the dual demands of high decoding accuracy and ultra-low latency. As quantum processors operate at extremely high speeds, decoders are typically required to process syndromes within less than 1\,\textmu s in superconducting circuits, posing significant challenges for both decoding algorithm design and hardware implementation.
Among various QEC codes, the surface code has emerged as a leading candidate due to its high threshold and local stabilizer structure~\cite{fowler2012surface}. Within the surface code framework,
two prominent decoding strategies have been extensively adopted:  Minimum-Weight Perfect Matching (MWPM)~\cite{dennis2002topological} and Union-Find (UF)~\cite{delfosse2021almost}. 
MWPM achieves high accuracy by solving an optimal matching problem over the syndrome graph, but its high computational cost and inherent sequential nature of computation limit its processing speed. 
In contrast, the UF decoder achieves lower latency by clustering defects in parallel yet sacrifices accuracy under complex noise conditions.

To address these challenges, this work proposes an algorithm--hardware co-design for accurate, low-latency, and flexible QEC decoding.
At the algorithmic level, we introduce coset ensemble decoding
that enhances traditional UF decoding by considering logical-equivalent
cosets, the group of physical errors with the same syndrome
and logical effect.
At the hardware level, we develop a customized architecture tailored to the proposed
algorithm. In contrast to prior work~\cite{wu2025micro,liyanage2023scalable}, which relies
on spatial architectures whose hardware resources grow
with code distance, our hardware features high generality
and efficiency through temporal reuse and targeted pipeline
optimizations.
Together, these choices improve the accuracy--resource trade-off while enabling tunable decoding accuracy for different deployment needs.



This paper makes the following main contributions. The first improves decoding accuracy, while the latter two reduce decoding cost and latency, and improve hardware efficiency.

\begin{itemize} [leftmargin=*]
  \item A novel \emph{coset ensemble decoding algorithm} that leverages logically equivalent cosets to approximate coset-level maximum-likelihood decoding. This is achieved by ensemble forest exploration, which effectively improves the algorithmic performance.

  \item Several \emph{algorithmic optimizations}, including reverse-order elimination and lossless graph compression, together with an optimality analysis of the proposed algorithm, to further reduce computational complexity.

  \item A \emph{domain-specific hardware architecture} tailored for coset ensemble decoding. It achieves high generality and efficiency through temporal reuse and targeted pipeline optimizations. Several hardware optimizations, such as multi-bank memory hashing and hierarchical ID mapping, are proposed to further improve hardware efficiency and reduce latency.
\end{itemize}

\section{Background and Motivation}\label{sec:motivation}

\subsection{Background}

\subsubsection{Quantum Error Correction (QEC)}
\label{sec:background:qec}


QEC employs protocols to protect quantum information from decoherence and noise~\cite{Roffe2019}. An $[[n, k, d]]$ QEC code encodes $k$ logical qubits into $n$ physical qubits ($n>k$), overcoming fundamental constraints, like the No-Cloning Theorem and projective measurements' destructive nature, through non-local encoding and indirect syndrome extraction using ancilla qubits. The code distance $d$, defined as the minimum number of physical qubits supporting an undetectable logical error, quantifies error resilience. The logical error rate can be exponentially suppressed by increasing $d$ if the physical error rate is below a threshold.

\subsubsection{Stabilizer Code}
\label{sec:background:stabilizer}

An $[[n,k]]$ stabilizer code has stabilizer group $\mathcal S\subset \mathcal P_n$ generated by $m=n-k$ independent commuting Pauli operators. One may choose a canonical generating set $\{S_g, T_g, \bar X_j, \bar Z_j\}$~\cite{poulin2008iterativedecodingsparsequantum}, where $S_g$ are stabilizer generators, $\bar X_j,\bar Z_j$ are logical Paulis for the $j$-th logical qubit, and $T_g$ are \emph{pure errors} that generate the pure error group $\mathcal T$ and satisfy $\{T_g,S_g\}=0$ and $[T_g,S_h]=0$ for $h\neq g$.
Moreover, there exists an $n$-qubit Clifford unitary (an encoding circuit) $U$ such that
\begin{alignat}{9}
S_g &= U Z_g U^\dagger, \ & T_g &= U X_g U^\dagger, \quad &&g=1,\cdots,m,\\
\bar Z_j &= U Z_{m+j} U^\dagger, \ & \bar X_j &= U X_{m+j} U^\dagger, \quad &&j=1,\cdots,k,
\end{alignat}

where $Z_g$ and $X_g$ represent Pauli-$Z$ and Pauli-$X$ applied on the $g$-th qubit.
The commutation relation between $T_g$ and $S_{g^{\prime}}$ is given by

\begin{equation}
    S_g\cdot T_{g^{\prime}} =UZ_gU^{\dagger}UX_{g^{\prime}}U^{\dagger} = (-1)^{\delta_{g,g^{\prime}}}T_{g^{\prime}}\cdot S_g
\end{equation}
The logical operators satisfy the Pauli commutation relation $\bar{X}_j\cdot \bar{Z}_{j^{\prime}}=(-1)^{\delta_{j,j^{\prime}}}\bar{Z}_{j^{\prime}}\cdot\bar{X}_j$, and a code word from a stabilizer code space is a state of the form $U(\ket{0}^{\otimes m}\otimes\ket{\psi})$ where $\ket{\psi}\in (\mathbb{C}^2)^{\otimes k}$~\cite{poulin2008iterativedecodingsparsequantum}.

\subsubsection{Surface Code}
\label{sec:background:surface}
Surface code is a promising stabilizer code known for its high error threshold and compatibility with 2D lattice architectures with nearest-neighbor interactions~\cite{dennis2002topological}. As illustrated in \figref{fig:background:surface code}, a single logical qubit is encoded in a $d \times d$ array of physical (data and ancilla) qubits.
A key property is that a data qubit error anti-commutes with adjacent stabilizers of the opposite type: an $X$ error flips two neighboring $Z$ stabilizers, and a $Z$ error flips $X$ stabilizers. This results in non-trivial syndrome signals only at error chain endpoints, crucial for decoding.  To counter complex errors under circuit-level noise~\cite{surfacecodebelowthreshold2024}, decoding employs $d$ syndrome rounds, XOR-ing consecutive syndrome outputs to form detectors for identifying measurement errors~\cite{fowler2012surface}, isolating measurement errors from data qubit errors. Based on this, a 3D graph $G(\mathcal{V}, \mathcal{E})$, where vertices are detector events and edges are potential errors, can be constructed for decoding tasks~\cite{Higgott2025}.

\begin{figure}
    \centering
    \includegraphics[width=1.0\linewidth]{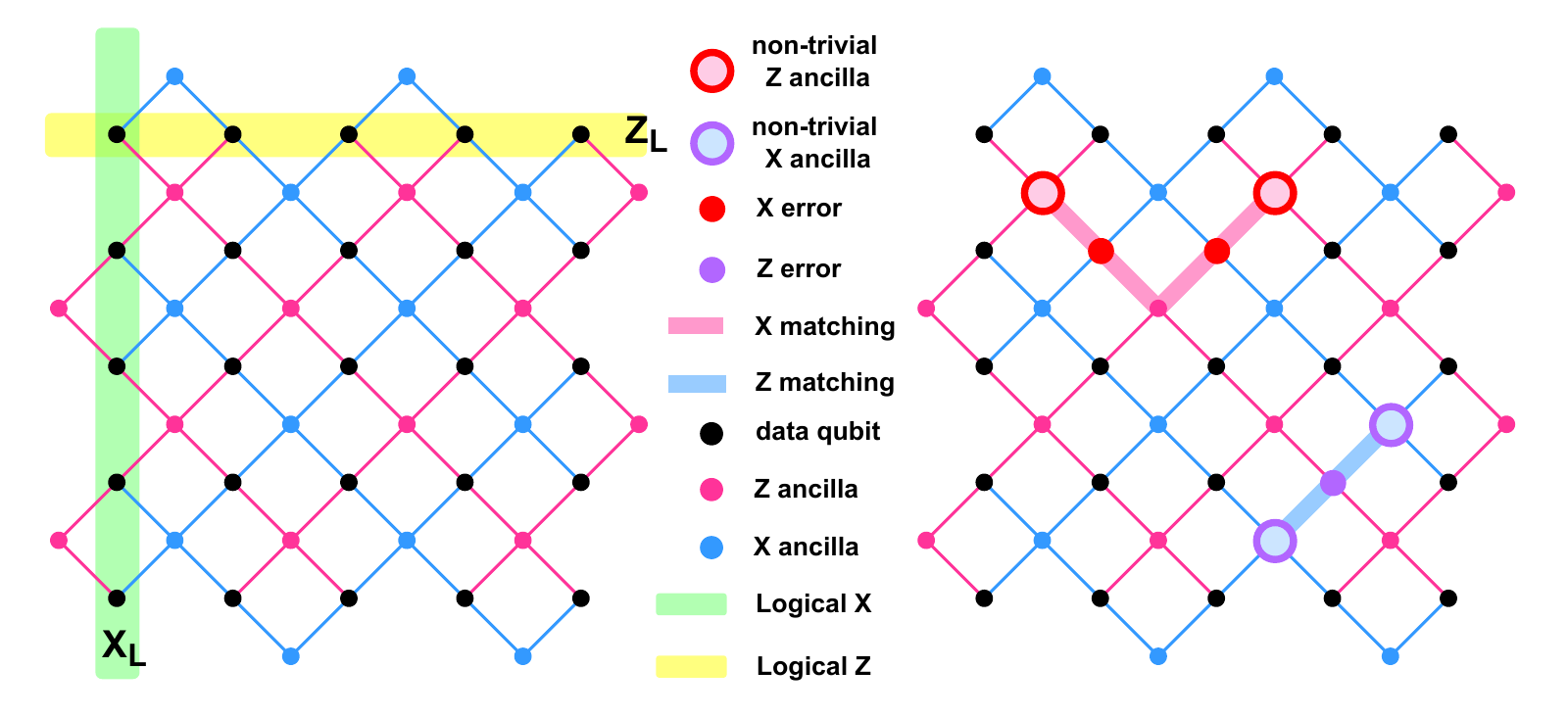}
    \caption{Layout of surface code (left) and syndrome behavior with corresponding matching for $X$ and $Z$ errors on data qubits (right).}
    \label{fig:background:surface code}
    \Description{}
\end{figure}

\subsubsection{Degeneracy}
\label{sec:background:degen}
Stabilizer codes, like Classical Error Correction (CEC), have the property that multiple error patterns can produce the same syndrome measurement~\cite{gottesman1997stabilizercodesquantumerror}. However, different QEC errors may be physically indistinguishable.
In QEC, any error operator $E$ can be decomposed into three components reflecting its stabilizer, pure error, and logical parts respectively as~\cite{poulin2008iterativedecodingsparsequantum}:
\begin{equation}\label{eq:expand of E}
E = s(E)\cdot t(s)\cdot l(E)=(s(E)l(E))\cdot t(s)
\end{equation}
where $s(E)\in\mathcal{S}$, $l(E)\in \mathcal{L}$, $s(E)l(E)\in \mathcal{N}(\mathcal{S})$. Pure error term $t(s)$ depends only on syndrome $s$: $t(s)=\prod_{g=1}^mT_g^{\frac{1-s_g}{2}}$, where $s_g \in \{+1,-1\}$ is the measured eigenvalue of the $g$-th syndrome. Thus, the same syndrome can correspond to distinct error patterns with same (differing only in $s(E)$) or different (differing in both $s(E)$ and $l(E)$) logical errors. Consider the first case, where errors differ only in $s(E)$. This happens if $E_1 = E_2 \cdot S$ with $S \in \mathcal{S}$. Under this condition, $E_1\ket{\psi}=E_2S\ket{\psi}=E_2\ket{\psi}$, and the same correction $R$ reverses both: $RE_2\ket{\psi}=\ket{\psi}\Rightarrow RE_1\ket{\psi}=RE_2S\ket{\psi}=RE_2\ket{\psi}=\ket{\psi}$. Such cases define indistinguishable degenerate errors, differing in physical error pattern but sharing both syndrome and logical effect, that can be grouped into the same logical-equivalent coset defined as $\{E|E=S_gt(s)L,\forall S_g\in\mathcal{S}\}$, which significantly impact decoding~\cite{9456887} (explained in Sec. ~\ref{sec:background:decoding}).

\subsubsection{Optimal Decoding (Coset vs. Physical)}
\label{sec:background:decoding}
Similar to CEC, it is straightforward to find the most probable physical error from syndrome $s$ via Maximum-Likelihood (ML) decoding:
\begin{equation}
    E^{*}=\argmax_{E\in\mathcal{P}_n}\{p(E|s)\}
\end{equation}
However, degeneracy implies that the most likely physical error may not yield the most probable logical error, since multiple equally probable errors can have the same logical effect~\cite{Stace2010}. For example (just for illustration), as shown in \figref{fig:background:coset example}, with non-trivial syndromes $\{1,2,3,4\}$, error patterns can be matched as $\{1,4\}\cup\{2,3\}$ or $\{1,2\}\cup\{3,4\}$, each implying a different logical error. Both matchings contain 6 errors (edges) and therefore have the same weight, which means they are equally probable. However, the second matching includes 9 combinations with equivalent probability, so the total probability of the second logical error would be higher than that of the first. Since the two matchings belong to different logical cosets, and the second has higher probability, we choose an error pattern from the second coset---even though a pattern in the first also has maximal individual likelihood. Thus, rather than identifying the most likely physical error, optimal decoding should find the most likely logical coset~\cite{poulin2008iterativedecodingsparsequantum}:
\begin{equation}\label{equation:coset ML}
        \begin{aligned}
            L^{*} &= \argmax_L\{p(L|s)\} \\
            p(L|s) &\propto \sum_{E:l(E)=L}p(E)= \sum_{S\in \mathcal{S}}p(E=St(s)L)
        \end{aligned}
\end{equation}

\begin{figure}
    \centering
    \includegraphics[width=210pt]{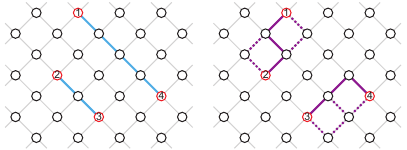}
    \caption{Two matching examples: the blue (left) belongs to coset 1 with logical error $L_b$, the purple (right) to coset 2 with logical error $L_p$. Dashed lines indicate alternative error paths.}
    \label{fig:background:coset example}
    \Description{}
\end{figure}

\subsubsection{Prior Decoding Algorithms}
\label{sec:background:decodingalgorithm}
The two main surface code decoders are Minimum-Weight Perfect Matching (MWPM) and Union-Find (UF). MWPM addresses $\argmax\{p(E|s)\}$ by solving the minimum-weight perfect matching problem on the decoding graph via the Blossom algorithm~\cite{wu2025micro}, which formulates the problem in a linear programming (LP) framework~\cite{wu2023fusionblossomfastmwpm}. 
While this LP-based iterative procedure yields near-optimal decoding accuracy by maintaining \textit{Blossoms}, it incurs substantial algorithmic and implementation complexity~\cite{Kolmogorov2009BlossomVA}, as well as relatively large decoding latency for real-time scenarios. In contrast, UF, an approximation of MWPM, is a simplified and near-linear-time algorithm~\cite{delfosse2021almost} that clusters syndromes by growing and merging them until all clusters have even parity, then resolves matches via spanning tree generation and peeling. While UF offers much lower latency due to its simplicity and high parallelism~\cite{das2022afs,Liyanage2024}, this comes at the cost of suboptimal accuracy compared to MWPM.

\subsection{Motivation}

Prior hardware decoders (i) neglect the impact of degeneracy~\cite{Higgott2025,delfosse2021almost}, and (ii) often rely on highly customized designs to achieve ultra-low latency, resulting in poor scalability~\cite{Liyanage2024}.
Our work revisits both the decoding theory and the hardware architecture, leading to a coset-ensemble algorithm and a stall-resilient pipelined architecture.

\subsubsection{\textbf{Algorithmic Challenge: Coset vs.\ Physical}}

Prior decoders implicitly frame decoding as a physical ML problem: given a syndrome, they select the most likely physical error chain. In highly degenerate stabilizer codes, many distinct chains belong to the same logical coset, and the single most likely chain does not necessarily correspond to the most probable logical error. As a result, the decoder is optimized for recovering a particular localized configuration, while system-level reliability is determined by the logical error rate.

\noindent\textbf{Insight and Our Approach.}
Given a syndrome, the ideal decision selects the coset with the highest posterior probability by summing over all compatible physical errors. This coset maximum-likelihood rule directly optimizes the logical error rate but may have exponential complexity in the code size~\cite{Hsieh2011}.

Inspired by the above insight, we approximate the intractable coset ML issue with our \emph{Ensemble Forest Exploration}. By using UF-equivalent \textit{Clustering} to partition intractable cosets and independently injecting random priority during the forest construction, the decoder generates a set of candidate corrections that implicitly ``vote'' for different logical cosets. Aggregating these outcomes at the logical level yields an approximate ranking of cosets in polynomial time.

\subsubsection{\textbf{Architectural Challenge: Scalability vs.\ Latency}}
Grid-mapped hardware decoders achieve extremely low latency but suffer from poor scalability and efficiency. A more scalable Von Neumann-style organization naturally decomposes the decoder into two stages: a \emph{pipelined clustering engine} that streams one vertex per cycle, followed by \emph{post-clustering traversal modules} (e.g., spanning-tree construction and peeling in conventional UF) that run once clustering terminates. Profiling this two-stage baseline in \figref{fig:motivation:stall} reveals that the dominant bottleneck lies inside the clustering pipeline: stalls in the clustering pipeline alone account for 48\%--58\% of the total decoding latency.

\noindent\textbf{Insight and Our Approach.}
Targeting an efficient decoder, we observe that the vast majority of stalls arise from two dominant patterns of concurrent memory accesses during cluster growth and merging stages.

Guided by this observation, we co-design a conflict-aware multi-bank hashed memory system and a hierarchical ID mapping scheme to resolve these two classes of conflicts and keep the pipeline highly utilized.

\begin{figure}
    \centering
    \includegraphics[width=0.95\linewidth]{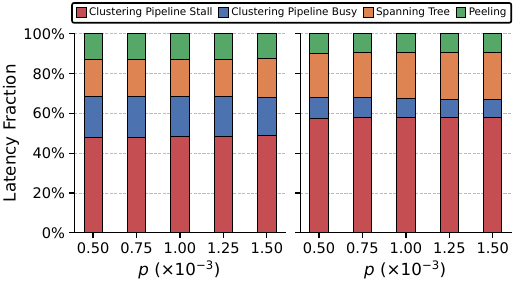}
    \caption{Latency breakdown for code distances 3 (left) and 11 (right) on a
    two-stage baseline without optimizations.
    Clustering Pipeline Stall and Clustering Pipeline Busy are cycles
    spent inside the \emph{pipelined clustering engine} (stalled vs.\
    productive); Spanning Tree and Peeling are cycles spent in the
    \emph{post-clustering traversal modules}.}
    \label{fig:motivation:stall}
    \Description{}
\end{figure}

\section{Coset Ensemble Decoder}

This section presents three algorithmic contributions that together define our coset ensemble decoding procedure (\alref{alg:overview}): \emph{ensemble forest exploration} (\secref{sec:meth:forest}), \emph{reverse-order elimination} (\secref{sec:meth:roe}), and \emph{lossless graph compression} (\secref{sec:meth:compression}).

\subsection{Ensemble Forest Exploration}
\label{sec:meth:forest}

\subsubsection{Algorithm Overview}

We adopt the coset viewpoint of stabilizer decoding: any error $E$ can be decomposed as $E=s(E)\,t(s)\,l(E)$, 
where the syndrome $s$ fixes $t(s)$ and ambiguity concentrates in the logical/coset choice.
To approach the coset-level maximum-likelihood decision $\argmax_L\{p(L|s)\}$, 
we sample a keyed \emph{priority function} $\phi$ over vertices and incident edges, construct one deterministic forest per priority sample, and obtain its candidate correction through linear-time ROE.
Repeating this process over $K$ independent samples yields coset-consistent candidates whose logical outcomes are aggregated by voting, 
as shown in \alref{alg:overview} and \alref{alg:priority_forest}.

\begin{algorithm}[t]
\caption{Coset Ensemble Decoder}
\label{alg:overview}
\small
\begin{algorithmic}[1]
\Require Syndrome parity $s$; decoding graph $G=(\mathcal{V},\mathcal{E})$; candidate number $K$; seeds
\Ensure Final correction $\hat{E}\subseteq \mathcal{E}$

\Statex \textbf{Phase I: Clustering}
\State $\hat{G} \gets \Call{Clustering}{G,s}$

\Statex \textbf{Phase II: Ensemble Forest Exploration}
\State $E \gets \varnothing$, $L \gets \varnothing$ 
\For{$i = 1$ \textbf{to} $K$}
  \ForAll{$(v,e) \in \hat{G}$}
    \State $\phi(v,e) \gets \mathrm{HashToUnit}(\mathrm{seed}, i, v, e)$
  \EndFor
  \State (\texttt{parent}, $\sigma$) $\gets$ \Call{PriorityForests}{$G, \phi$} \Comment{\alref{alg:priority_forest}}
  \State $\{E_i, L_i\} \gets \Call{ROE}{\texttt{parent}, \sigma, s}$ \Comment{\alref{alg:roe}}
  \State $E \gets E \cup \{E_i\}$; \quad $L \gets L \cup \{L_i\}$
\EndFor
\State $\hat{E} \gets \Call{MajorVote}{E,\,L}$ \Comment{on smallest-$|E_i|$ subset}
\State \Return $\hat{E}$
\end{algorithmic}
\end{algorithm}

\begin{algorithm}[t]
\caption{PriorityForests}
\label{alg:priority_forest}
\small
\begin{algorithmic}[1]
\Require Graph $G=(\mathcal{V},\mathcal{E})$; Priorities $\phi:\mathcal{V}, \mathcal{E} \to(0,1)$
\Ensure Array \texttt{parent}: $\mathcal{V}\to \mathcal{V}\cup\{\texttt{NIL}\}$; discovery order $\sigma$
\State \texttt{visited}$[v]\gets \textbf{false}$, \texttt{parent}$[v]\gets \texttt{NIL}$ for all $v\in \mathcal{V}$; \quad $\sigma\gets [\ ]$; $Q\gets$ \Call{Queue}{}
\State $\Pi_{\mathcal{V}} \gets \mathcal{V}$ sorted by ascending $\phi$
\ForAll{$u\in \Pi_{\mathcal{V}}$}
  \If{not \texttt{visited}$[u]$}
    \State  \Call{Enqueue}{$Q,u$}; \ \texttt{visited}$[u]\gets \textbf{true}$; \ \Call{Push}{$\sigma,u$}
    \While{not \Call{Empty}{$Q$}}
      \State $x\gets$ \Call{Dequeue}{$Q$}
      \State $\mathrm{Adj}\gets$ \Call{Adj}{$G,x$} sorted by ascending
      \ForAll{$y\in \mathrm{Adj}$}
        \If{not \texttt{visited}$[y]$}
          \State \texttt{parent}$[y]\gets x$; \ \texttt{visited}$[y]\gets \textbf{true}$; \ \Call{Enqueue}{$Q,y$}; \ \Call{Push}{$\sigma,y$}
        \EndIf
      \EndFor
    \EndWhile
  \EndIf
\EndFor
\State \Return (\texttt{parent}, $\sigma$)
\end{algorithmic}
\end{algorithm}

\begin{algorithm}[t]
\caption{Reverse-Order Elimination (ROE)}
\label{alg:roe}
\small
\begin{algorithmic}[1]
\Require Array \texttt{parent}; discovery order $\sigma$; parity $s\in\{0,1\}^{V}$
\Ensure Correction $E_{i}\subseteq \mathcal{E}$
\State $E_{i}\gets \emptyset$; \quad $p\gets s$
\For{$t=|\sigma|$ \textbf{down to} $1$}
  \State $x\gets \sigma_t$; \quad $r\gets$ \texttt{parent}$[x]$
  \If{$r\neq \texttt{NIL}$ \textbf{and} $p[x]=1$}
    \State $E_{i}\gets E_{i}\cup\{(x,r)\}$ 
    \State $p[x]\gets p[x]\oplus 1$; \quad $p[r]\gets p[r]\oplus 1$
  \EndIf
\EndFor
\State $L_i$ = \Call{DecodeLogical}{$E_i$}
\State \Return $\{E_{i}, L_{i}\}$
\end{algorithmic}
\end{algorithm}

\subsubsection{Proof of Approximation to Optimality}
As shown in Eq.~\ref{equation:coset ML}, the optimal ML decoding over logically equivalent cosets is computationally challenging. This is because the size of the Abelian group generated by the stabilizers grows exponentially with the number of stabilizers, and calculating the maximum-likelihood requires summation over this exponentially large group. Our approach approximates this optimal decoding by first partitioning the stabilizer group via clustering and then solving the sub-optimal coset decoding problem over the resulting partitioned cosets. To formalize this approximation, let's first introduce the following definitions and premises:

\begin{definition}[Syndrome graph]
    \label{definition:input graph}
    We define a \emph{syndrome graph} as an undirected connected graph $G(\mathcal V,\mathcal E)$ such that $\mathcal V$ can be partitioned into $\mathcal{V}_t$ and $\mathcal{V}_{nt}$ such that $\mathcal{V}_t = \{v |s(v)=+1 \}$ is called the set of trivial syndromes, $\mathcal{V}_{nt}=\{v|s(v)=-1\}$ the set of non-trivial syndromes, and the size of non-trivial syndromes is even. If the size of non-trivial syndromes is zero, we call the syndrome graph trivial. The input decoding graph of our algorithm is a non-trivial syndrome graph.

\end{definition}

\begin{definition}[Clustering]
    \label{definition:clustering}

    Clustering takes in the input decoding graph $G(\mathcal V,\mathcal E)$ and outputs a partition of sub-graphs $\mathcal C = \{G_i(\mathcal V_i,\mathcal E_i)\}$ such that each $G_i$ is a syndrome graph and precisely one $G_i$ is trivial.

\end{definition}

\begin{premise}\label{premise:before majorvote}
    After clustering, each non-trivial syndrome graph is sent to ensemble-forest-exploration. The $K$ independent priority samples induce $K$ error ensembles $\{E_i, L_i\}_{i=1}^K$, where $E_i\in\mathcal{P}_n$ and $L_i$ is the corresponding logical error.
\end{premise}

We now present two lemmas to substantiate the claim that our algorithm solves a sub-optimal coset ML problem. The first lemma demonstrates that the errors of candidate error ensembles are degenerate, forming logically equivalent cosets. The second lemma analyzes the algorithm's asymptotic optimality within the partitioned solution space after clustering.
\begin{lemma}\label{lemma:degeneracy}
    For $K$ error ensembles $\{E_i, L_i\}$, the $E_i$ with equal $L_i$ are degenerate errors and belong to the same logical equivalent coset of the stabilizer group $\mathcal{S}$.
\end{lemma}
\begin{proof}\label{proof:degeneracy lemma}
    As shown in Eq.~\ref{eq:expand of E}, $s(E)$ is a stabilizer term that could deform the error chain, or, equivalently, affect the matching in surface code. $t(s)$ corresponds to the pure error operator, which only depends on the syndrome measurement $s$, and $l(E)$ is the operator of the $E$'s logical error. Therefore, for any two error ensembles $\{E_1, L\}$ and $\{E_2, L\}$ with the same logical error, they both derive from the same syndrome measurement pattern, so their error expansions form the same $t(s)$ and $l(E)$. Because $L$, $S$, and $t(s)$ are all Hermitian and unitary,
    \begin{equation}
    E_1E_2^{\dagger}=\left(S_1t(s)L\right)\left(S_2t(s)L\right)^{\dagger}=S_1S_2
    \end{equation}
    Therefore, it's proved that for all error ensembles with the same logical error, their errors $E_i$ are degenerate and belong to the same logical-equivalent coset. 

\end{proof}

\begin{lemma}\label{lemma:sub-coset}
    The clustering $\mathcal{C}$ reformulates the global optimal coset decoding problem $\argmax_L\{p(L|s)\}$ by solving a locally optimal Maximum-Likelihood problem
    \begin{equation}
    \argmax_L\left\{\sum_{b\in\mathcal B_c}p(E=S(b)t(s)L)\right\}
    \end{equation}
    where for index $c$ running over all non-trivial syndrome graphs in the cluster,
    \begin{equation}\label{eq:287}
    \mathcal B_c = \{b\in\mathbb F_2^m | b_g=0\text{ for all } g \text{ with } s_g\notin G_c\}
    \end{equation}
    is the set of all $m$-bit strings whose support is restricted to the indices in the cluster.

    \rvs{with $b\in \mathbb{F}_2^m$, $S(b): \mathbb{F}_2^m\mapsto \mathcal{S}$, and 
    $\mathcal{B}=\bigsqcup_c\mathcal{B}_c\subset \mathbb{F}_2^m$, where $\mathcal{B}_c=\{b|b=\bigoplus_{g=1}^m\hat{b}\cdot\llbracket s_g\in\mathcal{G}_c\rrbracket\cdot (1\ll(g-1)), \ \hat{b}_g\in\{0,1\}\}$.
    Here, each $b$ is a $m$-bit bitstring, with its $g$-th bit $b_g$ equal to 0 or 1 if syndrome $s_g$ is contained in cluster $\mathcal{G}_c$, and only equal to 0 otherwise. Therefore, if $\mathcal{G}_c$ has $r$ syndromes (both trivial and non-trivial), the corresponding $\mathcal{B}_c$ would have $2^r$ elements.}{}
\end{lemma}

\begin{proof}\label{proof:sub-coset}
    Given a stabilizer code and the syndrome measurement result $s$, the optimal decoding is to find the logical-equivalent coset with the highest probability by solving Eq.~\ref{equation:coset ML}. Since the operator of each stabilizer generator $S_g$ can be constructed as $S_g = UZ_gU^{\dagger}$, and each stabilizer in $\mathcal{S}$ is the product of several stabilizer generators, the stabilizer $S$ can be written as:
    \begin{equation}
        \begin{aligned}
            S &= UZ_m^{b_m}U^{\dagger}UZ_{m-1}^{b_{m-1}}U^{\dagger}\cdots UZ_1^{b_1}U^{\dagger} \\
            \ &= U\left(Z_m^{b_m}Z_{m-1}^{b_{m-1}}\cdots Z_1^{b_1}\right)U^{\dagger}, \quad b_g\in\{0,1\}
        \end{aligned}
    \end{equation}
    If $b_g=1$, the $g$-th stabilizer generator is multiplied by the pure error term, which would cause a local deformation on the error chain. Therefore, if an $m$-bit bitstring $b$ is constructed by concatenating $b_g$ as $b=\bigoplus_{g=1}^m b_g\cdot(1\ll(g-1))$, the error operator will only depend on the value of this bitstring and can be rewritten as $E=S(b)t(s)L$. The activation of $g$-th bit of $b$ represents the contribution of $S_g$ on deforming the final error chain. In this case, the original coset probability becomes:
    \begin{equation}
        \sum_{S\in\mathcal{S}}p(E=St(s)L)\equiv \sum_{b\in\mathbb{F}_2^m}p(E=S(b)t(s)L)
    \end{equation}
    After clustering $\mathcal{C}$, the error chain could only be modified locally within each cluster. This implies that only the non-trivial (activated) syndromes within a cluster contribute to its error deformation\rvs{, while syndromes outside all clusters are deactivated}. Given the one-to-one correspondence between each bit $b_g$ of $b$, a syndrome bit $s_g$, and a stabilizer $S_g$, $b_g$ can vary (0 or 1) or deactivate (set to 0) if its corresponding syndrome $s_g$ is inside or outside any cluster $\mathcal{G}_c$. Each cluster thus defines a local configuration. The space of valid bitstrings is a subset of $\mathbb{F}_2^m$, which is the union of the spaces $\mathcal{B}_c$ as given in~\eqref{eq:287}.
    In this framework, the clustering $\mathcal{C}$ approximates the original global optimization problem from a locally optimal version by partitioning the stabilizer space into activated and deactivated regions\rvs{, which sets an optimality upper bound}.
\end{proof}

Based on the preceding lemmas, the error ensembles produced contain degenerate errors that can be grouped into logically equivalent cosets based on their logical errors. Under the priority-sampling distribution, the sample frequency of a logical outcome estimates its probability mass within the partitioned candidate space, which can be represented as $\tilde{p}(L_i|s)= \frac{n_{L_i}}{K}$, where $n_{L_i}$ is the number of ensembles with logical error $L_i$. The \Call{MajorVote}{} thus identifies the most frequently sampled coset and approximates the sub-optimal coset ML problem as
\begin{equation}
\label{eq:majorvote}
    \argmax_{L_i}\{\tilde{p}(L_i|s)\}= \argmax_{L_i}\left(\frac{n_{L_i}}{K}\right)
\end{equation}
A final correction can then be chosen arbitrarily from this coset due to the degeneracy among candidates. Moreover, as the candidate number $K\rightarrow\infty$, the sampling estimate converges within this partitioned candidate space, though performance remains bounded below the original optimal coset ML due to the clustering constraint. In practice, this vote is restricted to the candidates with the smallest correction size $|E_i|$. This empirical refinement improves accuracy and reduces to Eq.~\eqref{eq:majorvote} when all candidates share the same correction size.

\subsubsection{Relationship to UF and MWPM} 
Our coset ensemble decoder typically sits between UF and MWPM. It leverages UF's efficient clustering but critically advances it by introducing a coset-decoding step, which produces multiple error ensembles to identify the most probable logical coset and thus outperforms UF. Directly contrasting with MWPM reveals the impact of the clustering stage. MWPM performs maximum-likelihood (ML) decoding on physical errors, whereas our coset ensemble decoder solves a constrained version of the coset maximum-likelihood problem, and the solution space, a subset of all cosets defined by the clustering, is sub-optimal. Consequently, our decoder may exceed MWPM's accuracy only when the cluster structure aligns with the error structure that MWPM's \textit{Blossom} would capture.

\subsection{Reverse-Order Elimination (ROE)}
\label{sec:meth:roe}

Given the parent array and discovery order $\sigma$ from \secref{sec:meth:forest}, ROE scans vertices in the exact reverse of $\sigma$ and pops vertices in order as shown in \alref{alg:roe}.
This eliminates global leaf detection and degree recomputation, delivering a single-pass, linear-time peeling that helps reduce the decoding latency.

The key observation is that \alref{alg:priority_forest} has already traversed the graph
from roots to leaves during forest construction. By recording this order and
popping vertices in reverse, ROE reuses that traversal and avoids an additional
pass for leaf discovery.

\subsection{Lossless Graph Compression}
\label{sec:meth:compression}

\begin{figure}
    \centering
    \includegraphics[width=235pt]{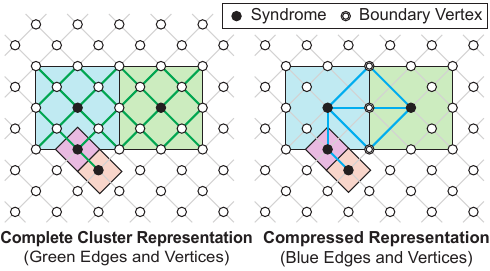} 
    \caption{Graph compression.}
    \label{fig:al:compress}
    \Description{}
\end{figure}

To reduce the additional exploration efforts introduced by \secref{sec:meth:forest}, we apply structure-preserving reductions with smaller complexity.
The example in \figref{fig:al:compress} illustrates how we obtain the compressed
graph structure after clustering. The four colored regions represent the
clusters grown from four initial root vertices. On the left, all vertices and
green edges inside the colored regions constitute the input graph
$G(\mathcal{V}, \mathcal{E})$ of \alref{alg:overview}. Since the complexity of \alref{alg:overview} is linear in
the size of the input graph, graph pruning that preserves its structural
information helps reduce the decoding latency.

In this work, we use the compressed graph structure shown on the right. Unlike
the complete graph on the left, we retain only the edges between roots and the
edges between roots and boundaries during merging. This edge representation,
which goes beyond axis-aligned Manhattan connections and allows edges to link
vertices in arbitrary directions across the grid, together with the pruning of
redundant edges, preserves the core structure of the graph while remaining
fully compatible with the main dataflow of Algorithm~1. In
\figref{fig:al:compress}, the straightforward representation uses an input graph
of size 21, whereas our compressed representation reduces this number to 8.

\section{Hardware Architecture}

\begin{figure*}
    \centering
\includegraphics[width=515.52pt]{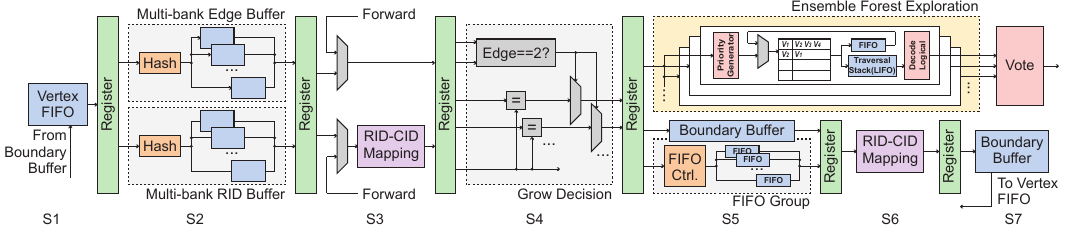} 
    \caption{Two-stage hardware architecture: a fully pipelined clustering engine (outside the light-yellow region) feeds the post-clustering modules (inside the light-yellow region), comprising $K$ parallel Ensemble Forest Exploration (EFE) instances and a Voting module that aggregates the $K$ candidate corrections into the final logical-error estimate.}
    \label{fig:hw:hwoverview}
    \Description{}
\end{figure*}

\subsection{Hardware Overview}

\figref{fig:hw:hwoverview} shows a two-stage architecture aligned with the algorithm: a fully pipelined clustering engine feeds $K$ parallel \emph{Ensemble Forest Exploration (EFE)} instances and a Voting module. Clustering is a streaming growth-and-merge computation that naturally maps to a seven-stage pipeline (S1--S7). Each EFE instance then performs a stateful forest traversal followed by ROE; because this traversal state cannot be time-multiplexed without overwriting in-flight adjacency data, the $K$ instances are replicated and run in parallel. Adjacency-list construction overlaps with clustering, and after clustering terminates each EFE instance traverses under a distinct priority scheme before the Voting module aggregates the predicted logical outcomes.

The main architectural bottleneck lies in the clustering pipeline: as profiled in \figref{fig:motivation:stall}, pipeline stalls dominate decoding latency. Clustering has low arithmetic intensity and is driven by indirect metadata accesses, so concurrent updates to shared global data can trigger severe bank conflicts. The remainder of this section therefore focuses on the conflict-aware multi-bank hashed memory system and hierarchical ID mapping scheme that keep this pipeline busy.

\subsection{Fully-Pipelined Architecture for Clustering Stage}
\label{sec:hw:fp}

These stalls primarily stem from RAW (Read-After-Write) hazards during sequential growth and memory bank conflicts under concurrent updates, motivating our specialized 7-stage pipeline in \figref{fig:hw:hwoverview} with targeted mechanisms to maximize hardware utilization.

The 7-stage pipeline processes one Vertex ID (VID) per cycle to achieve high-throughput decoding. To handle concurrent merge operations, we introduce a \textbf{hierarchical ID mapping} that interposes a Root-ID (RID) between VIDs and Cluster-IDs (CIDs). This indirection layer decouples physical vertex storage from logical cluster states. The dataflow operates as follows: \textbf{(S1--S3)} The pipeline dequeues a VID and concurrently fetches its associated RIDs and edge weights to resolve CIDs. \textbf{(S4)} Grow/merge logic is evaluated based on the retrieved metadata. \textbf{(S5--S7)} The pipeline manages active CIDs via a priority-based FIFO and updates boundary-vertex states. As detailed in \secref{sec:hw:remapping}, this RID-based hierarchy is key to collapsing write fan-out and mitigating peak memory bandwidth pressure.

To address the identified bottlenecks, we implement two core optimizations: (1) a \textbf{forwarding/bypass network} that feeds S4's growth decisions back to earlier stages to resolve data dependencies, and (2) a \textbf{hash-based memory layout} that minimizes contention during concurrent metadata lookups.

\subsection{Multi-Bank Memory Layout and Hashing Scheme}
\label{sec:hw:mb}

To handle highly concurrent local memory accesses during cluster growth, we design a customized multi-bank memory with conflict-free hashing. The hashing maps vertex and edge memory accesses to different memory banks, supporting single-cycle access for distances up to 15\footnote{Longer distances incur slightly higher clock cycle overhead.}.


Our memory system must satisfy a critical requirement: for any input lattice coordinate $(x,y,z)$, we need to simultaneously access (i) the center vertex and its axis-aligned neighbors, and (ii) all incident edges. 
The core design principle is to ensure that these concurrent accesses land in distinct banks as shown in~\figref{fig:multibank_hashing}. 

\begin{figure}[h]
  \centering
  \includegraphics[width=240pt]{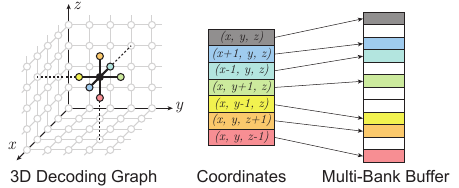}
  \caption{Multi-bank hashing distributes 7-vertex neighborhood to distinct banks.}
  \label{fig:multibank_hashing}
\end{figure}

We represent the 3D lattice using integer coordinates $(x,y,z)$
For edges, we use a consistent convention: each edge is identified by the coordinate of its \emph{positive (forward) endpoint}. Specifically:
\begin{equation}
\pm\mathbf{e}_i: (x,y,z) \leftrightarrow (x,y,z) \pm \mathbf{e}_i, \quad i \in \{x,y,z\}
\end{equation}
are all represented by $(x,y,z)$.
Our implementation supports lattices with code distance up to 15, with smaller configurations also supported.

To distribute vertex data uniformly across banks, we employ a linear congruential hash function that maps each vertex coordinate to a bank index $b_v$:
\begin{equation}
\begin{aligned}
b_v(x,y,z) &= (\alpha x + \beta y + \gamma z) \bmod M, \\
\end{aligned}
\end{equation}
where $\alpha=1,\beta=3,\gamma=5$, and $M=22$. 
A key property of these coefficients is that the resulting bank indices for the center $(x,y,z)$ and its axis-aligned neighbors 
are \emph{pairwise distinct by construction}, guaranteeing conflict-free concurrent accesses.

To support this bank distribution, vertices belonging to the same bank are densely packed using a lexicographic traversal order $(i,j,k)$ with $i$ outermost, then $j$, then $k$, where $i=0,\ldots,L-1$, $j=0,\ldots,L-1$, $k=0,\ldots,R-1$. The bank-internal address $a_v$ of a vertex is the \emph{rank} of $(x,y,z)$ among all triples that hash to the same bank:
\begin{equation}
\label{eq:av-rank}
\begin{aligned}
a_v(x,y,z)
=\sum_{\substack{0\le i, j<L\\0\le k<R}}
&[[(i+3j+5k)\bmod 22 = b_v(x,y,z)]]
\cdot \\
&[[(i,j,k)\prec_{\text{lex}} (x,y,z)]]
\end{aligned}
\end{equation}
where $[[\cdot]]$ is the Iverson bracket and $(i,j,k)\prec_{\text{lex}}(x,y,z)$ denotes lexicographic precedence. 

\subsection{Hierarchical ID Mapping for Cluster Merging}
\label{sec:hw:remapping}

Another source of memory conflicts happens during cluster merging. It features higher concurrency and weaker spatial locality. Consider the conventional merging process on a 3D grid with \(O(N^{3})\) points, where each coordinate stores a CID. During cluster growth, only a small, spatially contiguous set of points changes its CID. Our existing multi-bank buffer with hash-based placement handles these localized updates efficiently. In contrast, merges between concurrently growing clusters induce many logically simultaneous CID updates that may be scattered across the volume. Directly rewriting the ``VID\(\rightarrow\)CID'' store for tens of randomly located VIDs per merge scales poorly in hardware, amplifying both bank conflicts and write bandwidth demands.
\figref{fig:al:remapping} (a) presents a straightforward way of changing storage mapping relationships. Each thin line represents a storage mapping relationship. Under the straightforward approach, all storage cells whose CID was originally 3, 6, or 7 must be remapped to 1 during cluster merging, resulting in a total of 15 storage cells that need to be updated.

We introduce an intermediate representation that decouples high-fan-out, poorly localized merge updates from the coordinate address space. VIDs are first mapped to RIDs in the multi-bank, hash-partitioned buffer (optimized for growth). 
A compact memory then holds an ``RID\(\rightarrow\)CID'' indirection, where CID is the post-merge cluster identifier. Merge operations update only this ``RID\(\rightarrow\)CID'' mapping: all elements formerly addressed by the merged RIDs are logically relabeled by modifying a small number of RID entries rather than rewriting the many coordinates that reference them.
Because ``VID\(\rightarrow\)RID'' is already a many-to-one mapping during growth, the merge stage's write fan-out collapses from ``number of touched VIDs'' to ``number of touched RIDs,'' enabling single-cycle remaps in the specialized memory and sharply reducing peak concurrent traffic.
In \figref{fig:al:remapping} (b), we only update the memory mapping from RID to CID. In this example, the memory cells storing mapping relationships of RID 3, 6, and 7 are updated. Compared with the straightforward method, the concurrent memory access pressure has been relieved.

\begin{figure}
    \centering
    \includegraphics[width=240pt]{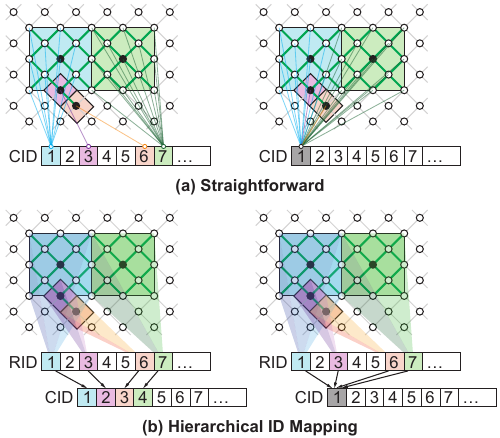} 
    \caption{Comparison of memory-cell update counts between the straightforward method and our method (left: before merging, right: after merging).}
    \label{fig:al:remapping}
    \Description{}
\end{figure}

\section{Experimental Methodology}

\subsection{Experimental Setup}\label{sec:exp:setup}

We evaluate the performance of our decoder from a comprehensive perspective, including accuracy, latency, and hardware efficiency.
The proposed hardware design is implemented in SystemVerilog HDL on a Xilinx Virtex UltraScale+ VU19P FPGA. The hardware resources and frequency are reported from Vivado~2024.2.
The algorithm performance evaluation is conducted through a Python-based hardware simulator, which is cross-validated against our hardware design.
It reports logical error rates and cycle counts, and tracks memory-access conflicts under our multi-bank memory layout and hashing scheme. 

Our experiments adopt several widely-used noise models to illustrate the generality of our decoder.
(1) Circuit-level depolarizing noise model implemented using the Stim library~\cite{gidney2021stim}. 
For a given code with distance $d$ and a specified number of syndrome extraction rounds, we generate noisy circuits in which depolarizing noise with rate $p$ is applied to data qubits after Clifford operations and between successive rounds of the circuit. Measurement errors are modeled as classical bit flips on the measurement outcomes with the same probability ($p$), while qubit reset operations are assumed to be ideal. Unless otherwise specified, we use $q=p$ and set the number of repeated syndrome rounds to $T=d$.
(2) Biased and unbiased Phenomenological noise model. For biased phenomenological noise, X- and Z-type data faults are injected with probabilities $p_X$ and $p_Z$, respectively, with bias ratio $\eta = p_Z/p_X$; measurement faults follow the same phenomenological model as above.

All algorithmic accuracy results in this paper are obtained on a surface code with periodic boundary conditions, the same setting used by QUEKUF~\cite{valentino2025quekuf}. For Micro-Blossom~\cite{wu2025micro} and Helios~\cite{liyanage2023scalable}, hardware-resource numbers are taken from their original publications on the rotated variant, while decoding latencies are reproduced by running their source code under matched noise conditions. Surface-code variants share the same threshold and differ only in boundary conditions~\cite{wang2009threshold}. Reproducing these baselines on a periodic-boundary surface code would increase their decoding latency, since the corresponding syndrome graph is larger; the reported values therefore provide a best-case estimate of these baselines.

\begin{figure*}[!t]
    \centering
    \includegraphics[width=1.0\textwidth]{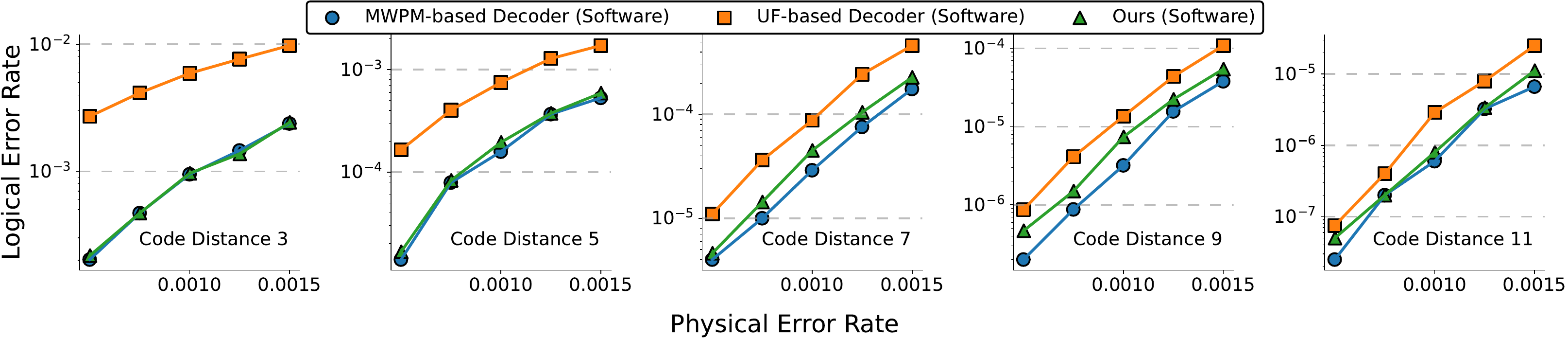}
    \caption{Logical error rate comparison among MWPM-based decoders, UF-based decoders, and our decoder.}
    \label{fig:LER}
    \Description{}
\end{figure*}

\subsection{Decoder Performance Metric}\label{sec:decoder metric}

\subsubsection{Real-time Compliance}

Modern quantum--classical systems impose tight decoding-latency constraints to prevent backlog, which would otherwise compromise logical fidelity and stall program execution.
Prior architecture works for superconducting platforms commonly target sub-microsecond decoding~\cite{holmes2020nisq,das2022afs,wu2025micro}.
Following these works, the real-time compliance of hardware decoders is set to the time of one syndrome extraction round.

\subsubsection{System Infidelity}
Decoding in fault-tolerant quantum computing (FTQC) can be broadly categorized into two types:

\begin{enumerate}
    \item \textbf{Pauli-frame decoding}: the decoding outcome is used solely to correct the measurement result of the corresponding logical qubit through Pauli frame updates. This is typical in memory experiments.
    \item \textbf{Feedback decoding}: the decoding result not only corrects the measurement of a logical qubit but also serves as feedback to conditionally apply logical operations on other qubits, common in implementing non-Clifford operations.
\end{enumerate}

To evaluate Pauli-frame decoding, metrics such as the \textit{logical error rate (LER)} and \textit{reaction time (latency)}~\cite{surfacecodebelowthreshold2024} are generally sufficient. However, in \textit{Feedback decoding}, the combination of decoding latency and accuracy becomes critical. For instance, suppose a logical operation on logical qubit $A$ is conditioned on the outcome of a \( Z \)-basis measurement on logical qubit $B$. If the decoding of $B$ takes \( R \) rounds of syndrome measurements (measured in cycles), then the physical qubits of $A$ must remain idle during this time. As a result, \( R \) additional rounds of memory decoding must be applied to $A$ before the conditional operation can proceed. However, this waiting increases the total logical error rate since more errors would be accumulated on physical qubits as discussed in~\cite{bausch2023learning}. Therefore, a fairer and more appropriate metric is required to evaluate the impact of decoding latency on decoding accuracy in feedback-based logical operations.
We defined this metric to quantify how the decoding latency of logical patch B affects the decoding fidelity of logical patch A, specifically when A’s operation is conditioned on B’s mid-circuit measurement result.

In Ref.~\cite{bausch2023learning}, the decoder error rate $E(n)$ after $n$ rounds of syndrome measurements assuming a per-round logical error rate $\epsilon$ is given by an empirical formula $E(n) = \frac{1}{2}(1 - (1 - 2\epsilon)^n)$. However, $\epsilon$ is not directly measurable in FTQC, where the fundamental unit is a full QEC cycle of $d$ syndrome rounds. We therefore reparametrize $E(n)$ in terms of the decoder's LER over $d$ rounds, $E(d)$.
Using \((1-2\epsilon)^d=1-2E(d)\) and \((1-2\epsilon)^n=((1-2\epsilon)^d)^m\) with $m=n/d$ the number of decoding cycles, we obtain the \textit{effective decoder error rate}
\begin{equation}
    \hat{E}(m) = \tfrac12\bigl(1 - (1-2E(d))^m\bigr)
\end{equation}
and the corresponding \textit{effective decoder fidelity}
\begin{equation}
    \hat{F}(m) = 1 - 2\hat{E}(m) = (1-2E(d))^m.
\end{equation}

In a feedback-decoding scenario, if the decoding latency for qubit $B$ is $R$ (in units of syndrome cycles) and qubit $A$ has been idle for $m$ decoding cycles, the effective fidelity of $A$ under $B$'s latency becomes
\begin{equation}
    \hat{F}\left(m + \tfrac{R}{d}\right) = (1-2E(d))^{R/d}\cdot\hat{F}(m),
\end{equation}
where $R$ is computed from $B$'s decoding latency and $E(d)$ is $A$'s decoding LER.
The impact of $B$'s latency is thus captured by the factor $(1-2E(d))^{R/d}$. If the decoding latency is shorter than one syndrome cycle, no backlog occurs and the LER is unaffected, so $R$ is floored at $1$. Inverting for convention (lower is better), the resulting \textit{Infidelity factor} is
\begin{equation}
    \hat{C}(R) = 1-(1-2E(d))^{\frac{\max(1,R)}{d}} \in [0, 1),
\end{equation}
with $R = L/l$ where $L$ is the decoding latency and $l$ the duration of one syndrome extraction round; the mask $\max(1,R)$ implies that if $B$'s latency is less than one extraction round, its impact on $A$'s fidelity is negligible and $\hat C(R)$ is dominated by $E(d)$. This threshold is sufficient for FTQC because as long as decoding completes before the next syndrome is extracted, latency does not degrade the LER. A lower $\hat C(R)$ indicates higher fidelity under latency constraints. While sensitive to both idle time and idle error rates, the impacts of physical idle errors, along with optimizations such as Dynamic Decoupling and Pauli Twirling, are fully encapsulated by $E(d)$.

\section{End-to-End Evaluation}\label{sec:exp}

We first evaluate algorithmic accuracy, latency, and system-level impact, and then analyze the hardware cost and scalability of the proposed design.

\subsection{Decoding Accuracy Evaluation}\label{sec:exp:eva:accuracy}

The decoding accuracy of the proposed method is evaluated under the circuit-level
noise model described in~\secref{sec:exp:setup}. To isolate the algorithmic gain
introduced by the ensemble-forest method in~\secref{sec:meth:forest}, we first
compare against two widely used surface-code decoders, MWPM and UF, at the algorithm level as shown
in~\figref{fig:LER}. For the accuracy estimate of MWPM-based decoders, we use the
PyMatching implementation~\cite{higgott2021pymatchingpythonpackagedecoding}. For the UF-based decoders Helios
and QUEKUF, we evaluate accuracy using our own baseline UF software
implementation to avoid conflating decoder quality with minor differences in
boundary-condition handling across implementations.

In these experiments, the candidate number $K$ in our design
is fixed to 24.
When the code distance is small, the accuracy of our coset ensemble decoding is close to that of MWPM.
For larger code distances, the increased graph size causes this fixed $K$ to limit further accuracy improvements.

\begin{figure}[!tbp]
    \centering
    \includegraphics[width=0.85\linewidth]{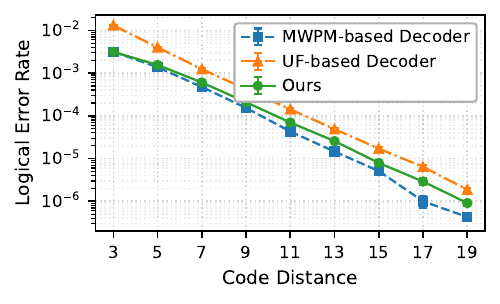}
    \caption{Logical error rate of MWPM, UF, and our decoder ($K{=}24$) vs.\ code distance at $p{=}0.002$ circuit-level noise.}
    \label{fig:LER_larger_distance}
    \Description{}
\end{figure}

\begin{figure*}[!t]
    \centering
    \includegraphics[width=1.0\textwidth]{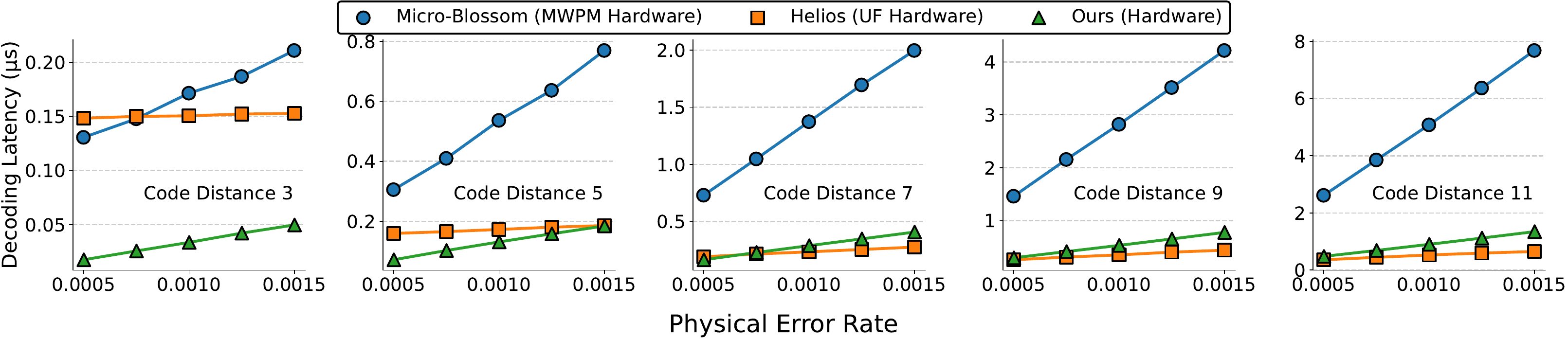}
    \caption{Decoding latency per decoding task ($d$ syndrome
    rounds), compared with state-of-the-art decoders.}
    \label{fig:latency}
    \Description{}
\end{figure*}

\begin{figure}[!tbp]
    \centering
    \includegraphics[width=1.0\linewidth]{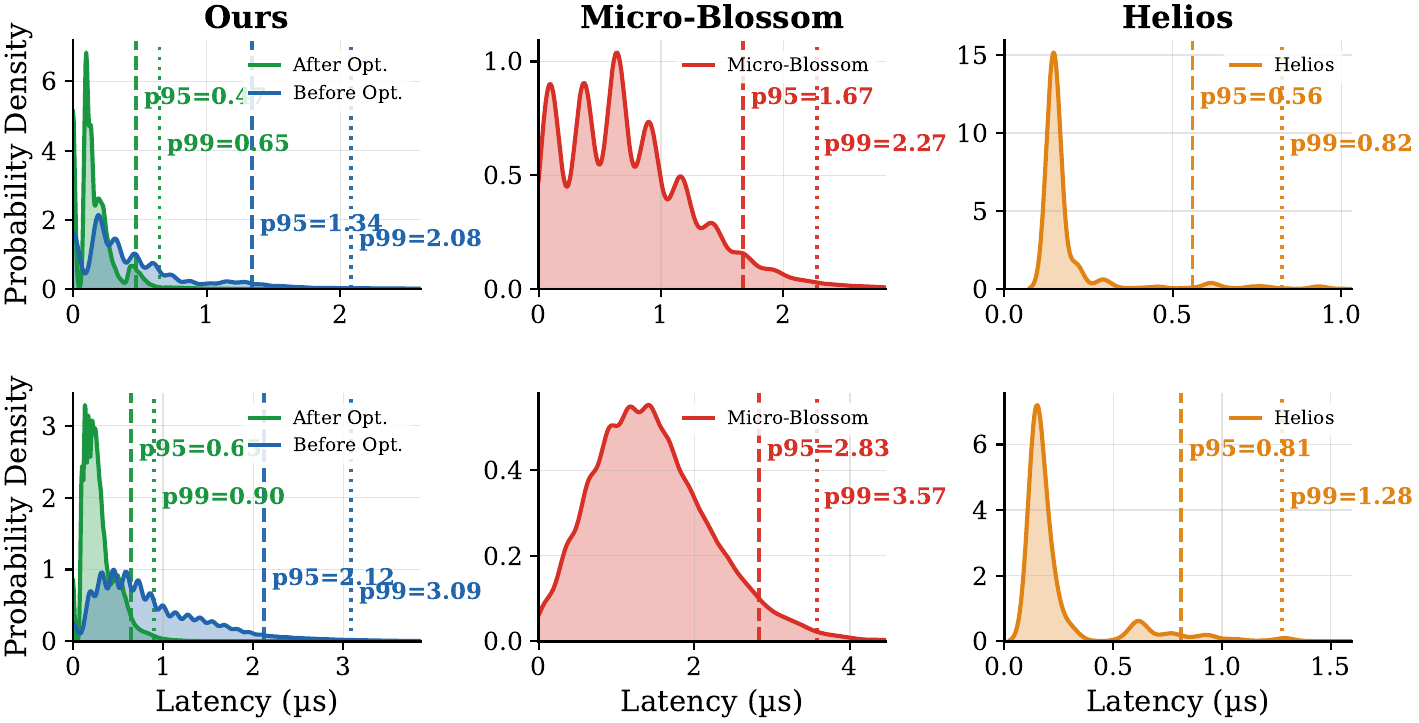}
    \caption{Decoding latency distribution per $d$-round task at
    $p{=}0.0005$, for code distance 7 (top) and 9 (bottom).}
    \label{fig:latencydistribnution}
    \Description{}
\end{figure}

\begin{figure*}[!t]
    \centering
    \includegraphics[width=1.0\textwidth]{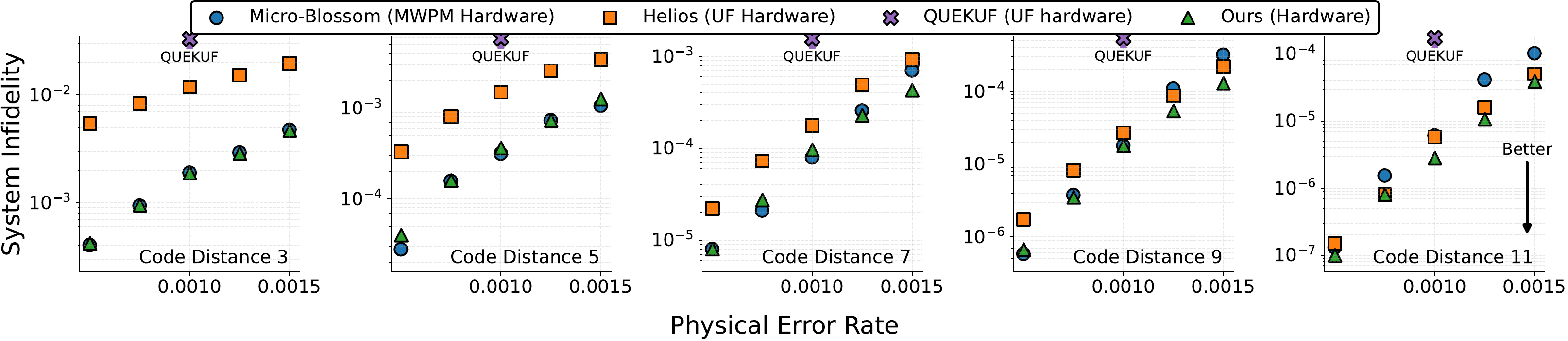}
    \caption{System infidelity comparison with SOTA decoders.}
    \label{fig:systemfidelity}
    \Description{}
\end{figure*}

\begin{figure}[!tbp]
    \centering
    \includegraphics[width=0.9\linewidth]{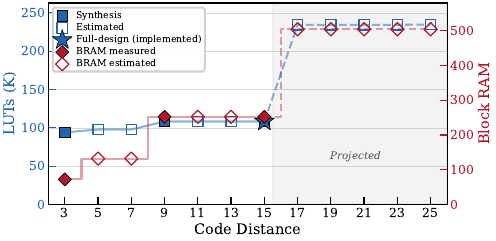}
    \caption{FPGA resource usage vs.\ code distance $d$. Filled markers denote full Vivado synthesis ($d{=}3,9,15$); open markers are estimates. The shaded region indicates extrapolation beyond measured data.}
    \label{fig:resource_scaling}
\end{figure}

\begin{figure*}[!t]
    \centering
    \includegraphics[width=1.0\textwidth]{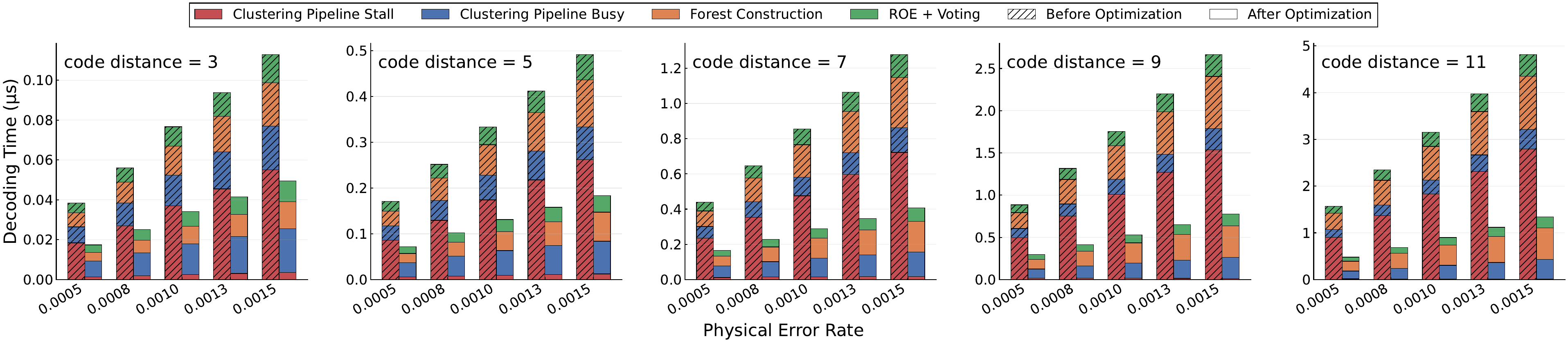}
    \caption{Hardware latency breakdown before and after optimization.}
    \label{fig:hwperfAnalysis}
    \Description{}
\end{figure*}

To examine how the accuracy advantage persists at larger
code distances, \figref{fig:LER_larger_distance} reports the logical error rate
of MWPM, UF, and our decoder at a fixed physical error rate
$p{=}0.002$ circuit-level noise for $d\in\{3,5,\dots,19\}$. Our
decoder tracks MWPM closely across this range, substantially
outperforming UF: the LER ratio to MWPM grows from $1.0\times$ at $d{=}3$ to ${\sim}2.1\times$ at $d{=}19$ at fixed $K{=}24$; the residual gap can be further reduced by increasing $K$.

\subsection{Decoding Latency Evaluation}

\figref{fig:latency} compares the average decoding latency per $d$-round task of our decoder against Micro-Blossom and Helios for $d=3$ to $d=11$.
The MWPM-based Micro-Blossom grows steeply with $d$, reflecting the high complexity of minimum-weight matching, while our decoder and Helios both remain sub-microsecond across all evaluated distances.

Helios's distributed per-vertex PE array yields a latency that scales sublinearly with $d$ through per-iteration coordination and convergecast, whereas our pipelined design scales with the active-vertex count while keeping the hardware footprint compact (the 24 ensemble candidates are processed in parallel and do not affect the critical path). At small $d$, where active vertices are few, our pipeline retires the task well below Helios's per-iteration floor, producing the $3$--$5\times$ advantage at $d{=}3$. As $d$ grows, the two curves converge near $d{=}7$--$9$, with our design still ahead in lower-$p$ settings at $d{=}7$; beyond this range, Helios's sublinear scaling becomes more favorable on pure latency. Our design targets a resource-efficient latency/area point, using roughly $6\times$ fewer LUTs and $3\times$ fewer FFs than Helios at $d{=}15$ (\tabref{tab:hardwareComp}).
Throughput follows the same pattern, ranging from $1.88$\,M decodes/sec at $d{=}9$ to $29.8$\,M at $d{=}3$ (both at $p{=}0.001$), which is $4$--$5\times$ Micro-Blossom and comparable to Helios.

\figref{fig:latencydistribnution} presents the latency probability density per full $d$-round decoding task. Our optimizations reduce not only average latency but also tail latency: at $d{=}9$, p95 drops from 2.12\,\textmu s to 0.65\,\textmu s ($-$69\%) and p99 from 3.09\,\textmu s to 0.90\,\textmu s ($-$71\%), with comparable improvements at $d{=}7$.
 The narrowed density curves illustrate that the optimizations reduce variance rather than merely shifting the mean. Compared with the Micro-Blossom, our optimized p99 is 3.5--4$\times$ lower (0.65 vs.\ 2.27\,\textmu s at code distance 7; 0.90 vs.\ 3.57\,\textmu s at code distance 9) and slightly below Helios (0.65 vs.\ 0.82\,\textmu s; 0.90 vs.\ 1.28\,\textmu s). 
 These results show that the proposed decoder achieves competitive or superior tail-latency behavior relative to state-of-the-art MWPM designs.

\subsection{System Infidelity Comparison}

\figref{fig:systemfidelity} compares our decoder with three baseline designs in terms of system infidelity, a more appropriate metric introduced in~\secref{sec:decoder metric}.
For small code distances \(d=3\), our decoder achieves essentially the same system infidelity as Micro-Blossom, and significantly outperforms both UF-based decoders. Although MWPM yields intrinsically lower LER than UF variants, Micro-Blossom's decoding latency exceeds the threshold at $d\!\ge\!5$ (\figref{fig:latency}), so $\hat C(R){>}0$ penalizes its system infidelity; it eventually becomes worse than that of Helios.
In contrast, our decoder maintains low infidelity thanks to its better scalability, achieving higher accuracy than the UF-based decoders while maintaining low latency.
At $d=11$, our decoder reduces the system infidelity by up to 74.3\% compared to Micro-Blossom and by 51.7\% compared to Helios.

\section{Hardware Performance Analyses}

\subsection{Matched/Normalized Hardware Resource Comparison}

Table~\ref{tab:hardwareComp} summarizes the hardware cost of our design and
three representative baseline decoders: Micro-Blossom~\cite{wu2025micro}, Helios~\cite{liyanage2023scalable}, and QUEKUF~\cite{valentino2025quekuf}.
To make a fair comparison, we further normalize prior results to a common code distance using the reported resources together with the scaling complexity described in the original papers. Specifically, we estimate the resource cost at the target distance by fitting/interpolating from the reported results under the stated scaling trend.
For latency comparisons, the clock frequency of each baseline is taken as published from its single design point, and is applied uniformly across all evaluated code distances; this matches our own setup, in which a single RTL design is used across all distances.
In terms of logic utilization, our architecture
requires only 108k LUTs, which is about $8.0\times$ fewer than
Micro-Blossom, and roughly $4.3\times$ fewer than
QUEKUF. A similar trend holds for flip-flops (FFs): our design uses 43k FFs, i.e.,
$2.9\times$ fewer than Helios and $14.7\times$ fewer than QUEKUF. 
Regarding on-chip memory, our design uses about half the BRAM of QUEKUF while supporting nearly twice the maximum code distance, which
demonstrates a better efficiency. In terms of achievable
frequency, our decoder runs at 163\,MHz, which is $3.7\times$ higher than Micro-Blossom and $2.1\times$ higher than Helios. 

We also quantify the dynamic-power overhead of ensemble parallelism using the Vivado power report. Each EFE branch contributes about 50\,mW of dynamic power, so the $K{=}24$ parallel branches together account for approximately 1.2\,W of dynamic power. 
Increasing $K$ therefore scales only the branch term linearly while leaving the shared clustering engine and voting part unchanged.

\begin{table}[htbp]
\centering
\caption{Hardware resource comparison. For fairer comparison across different code distances, Helios and QUEKUF are shown with both the originally reported values and normalized estimates at code distance 15.}
\label{tab:hardwareComp}
\small
\setlength{\tabcolsep}{2.8pt}
\renewcommand{\arraystretch}{1.1}
\begin{tabular}{lcccccc}
\toprule
& Micro-B. & \multicolumn{2}{c}{Helios} & \multicolumn{2}{c}{QUEKUF} & \textbf{Ours} \\
\cmidrule(lr){3-4} \cmidrule(lr){5-6}
&  & Orig. & Norm. & Orig. & Norm. &  \\
\midrule
LUT (k)    & 867   & 889   & 614   & 309   & 463  & 108 \\
FF (k)     & NA    & 177   & 126   & 453   & 634   & 43 \\
BRAM tiles     & 3    & NA    & NA    & 548   & 828   & 252 \\
Freq (MHz) & 43    & 75    & NA  & 238   & NA & 163 \\
Code distance    & 15    & 17    & 15    & 8     & 15          & 15 \\
\bottomrule
\end{tabular}

\vspace{2pt}
\end{table}

\subsection{Hardware Scalability with Code Distance}

Existing MWPM- and UF-based
hardware decoders typically exhibit rapidly growing resource consumption and
degraded clock frequency as the code distance increases. In contrast, our
architecture is explicitly designed to scale more efficiently with distance,
providing a more hardware-efficient solution toward larger-scale surface-code decoding.

\figref{fig:resource_scaling} estimates FPGA resource costs as $d$ scales from 3 to 25.
Across this sweep we scale only the components addressed by lattice coordinates---primarily the multi-bank vertex/edge buffer of the clustering engine and the per-EFE adjacency storage---while holding the rest of the design at its measured $d{=}15$ sizing.
The scaled buffers grow analytically as $O(2^{\lceil\log_2 d\rceil})$ due to power-of-two address quantization. 

\subsection{Latency Breakdown and Residual Stalls}

\figref{fig:hwperfAnalysis} breaks down the decoding latency under the same setting of \figref{fig:latency} before and after our optimizations. The clustering stage is pipelined so that its execution overlaps with the subsequent forest-construction and peeling stages; when the pipeline cannot be fed in time, the resulting stalls appear as idle cycles. In the baseline design, stalls in the clustering pipeline dominate execution time, accounting for 48\%--58\% of total latency across tested configurations. After applying the proposed optimizations, the clustering pipeline stalls fraction drops to 1--7\%, effectively eliminating pipeline bubbles. This translates to an overall speedup of 2.2--3.6$\times$ over the baseline, with larger gains at higher code distances (e.g., 3.0--3.4$\times$ at $d{=}9$ and 3.2--3.6$\times$ at $d{=}11$).
The residual stalls in the clustering pipeline originate from a parity-update RAW hazard in the clustering tail: once only a few active clusters remain in flight, a newly dispatched vertex cannot resolve its merge decision until the preceding growth step commits its parity update, briefly draining the pipeline between dispatches. Eliminating this 1--7\% residual would require either speculative parity evaluation or out-of-order dispatch, both of which incur non-negligible control-logic and memory-bandwidth overhead without affecting the logical error rate. We therefore retain the simpler in-order pipeline, which offers a favorable design point given the small remaining headroom.

\begin{figure}[!tbp]
    \centering
    \includegraphics[width=1.0\linewidth]{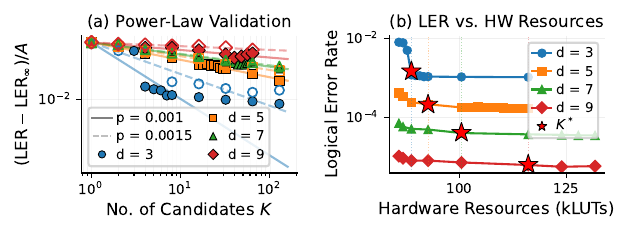}
    \caption{Tunability analysis of our proposal.
(a)~Power-law model validation at two noise rates ($p{=}0.001$, $0.0015$).
(b)~LER vs.\ hardware area.}
    \label{fig:tunability}
    \Description{}
\end{figure}

\section{Sensitivity and Robustness Analysis}

\subsection{Tunability Analysis}

We fit an empirical power-law model: $\mathrm{LER}(K) = \mathrm{LER}_\infty + A \cdot K^{-\alpha}$, with $\mathrm{LER}_\infty$ the error floor, $A$ the improvement headroom, and $\alpha$ the diminishing-returns exponent. \figref{fig:tunability}(a) shows the fit holds for each $(d, p)$, with $\alpha$ decreasing from $1.98$ ($d{=}3$) to $0.27$ ($d{=}9$). Defining $K^*$ as the smallest $K$ capturing $70\%$ of the LER improvement yields $K^* = 2^{\lfloor (d+1)/2 \rfloor}$. 
This remains accurate at $p{=}0.0015$ for $d \ge 7$, but underestimates $K^*$ for $d \in \{3, 5\}$, since larger $p$ requires more candidates per fractional gain. \figref{fig:tunability}(b) plots the accuracy--resource Pareto $\mathrm{LUT}_{\mathrm{total}} = \mathrm{LUT}_{\mathrm{fixed}} + K \cdot \mathrm{LUT}_{\mathrm{branch}}$; red stars mark $K^*$ near each knee.

\begin{figure}[!tbp]
    \centering
    \includegraphics[width=0.85\linewidth]{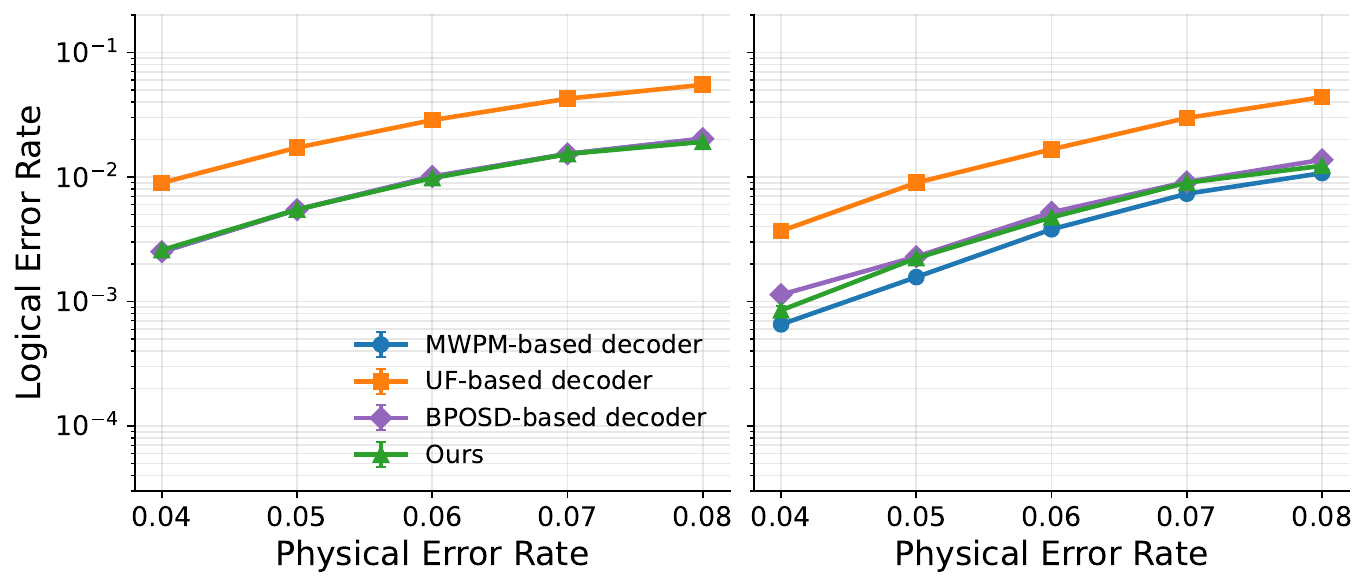}
    \caption{Logical error rate comparison on the repetition code under a phenomenological noise model.}
    \label{fig:LERrepetition}
    \Description{}
\end{figure} 

\begin{figure}[!tbp]
    \centering
    \includegraphics[width=0.95\linewidth]{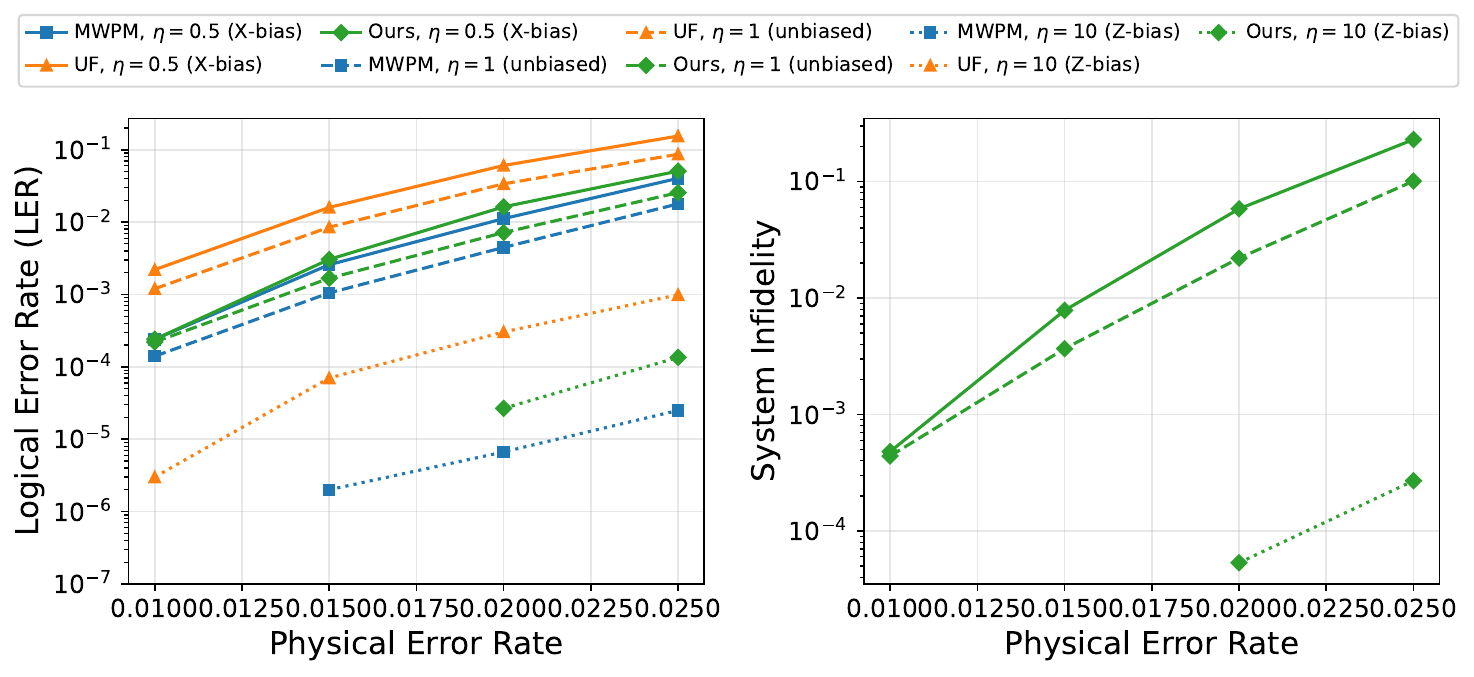}
    \caption{Logical error rate (left) and the corresponding system infidelity (right) under three biased phenomenological noise settings ($d{=}7$, X-channel).}
    \label{fig:biased_noise}
\end{figure}

\subsection{Compatibility and Comparison with BP+OSD}

To position our decoder on the accuracy spectrum between UF and near-optimal decoders, we include BP+OSD as an accuracy reference alongside MWPM and UF, using product-sum BP with OSD-CS of order 15 on the same Tanner graph for a fair comparison. The evaluation uses the repetition code under a phenomenological noise model, which additionally verifies the generality of our decoder across code families and noise models. We sweep code distances $d\in\{5,7\}$ and physical error rates $p\in[0.04,0.08]$.
\figref{fig:LERrepetition} reports the LER of MWPM, UF, BP+OSD, and our decoder, with the left and right panels corresponding to $d{=}5$ and $d{=}7$, respectively. Our decoder achieves LER within $1.0$--$1.4\times$ of MWPM, on par with BP+OSD ($1.0$--$1.7\times$), while UF trails by $2.7$--$5.7\times$. On these benchmarks, our decoder tracks MWPM and BP+OSD closely, decisively separating it from UF.

We also evaluate our decoder under biased phenomenological noise with three common bias ratios $\eta = p_Z/p_X$: $\eta{=}0.5$ (X-biased), $\eta{=}1$ (depolarizing), and $\eta{=}10$.
As shown in \figref{fig:biased_noise}, our decoder closes ${\sim}94\%$ of the UF-to-MWPM gap under X-biased noise, where vanilla UF incurs $6.2{\times}$ higher LER than MWPM.
The corresponding system-infidelity curves are shown in the right panel of \figref{fig:biased_noise}.

\begin{figure}[!tbp]
  \centering
  \includegraphics[width=\linewidth]{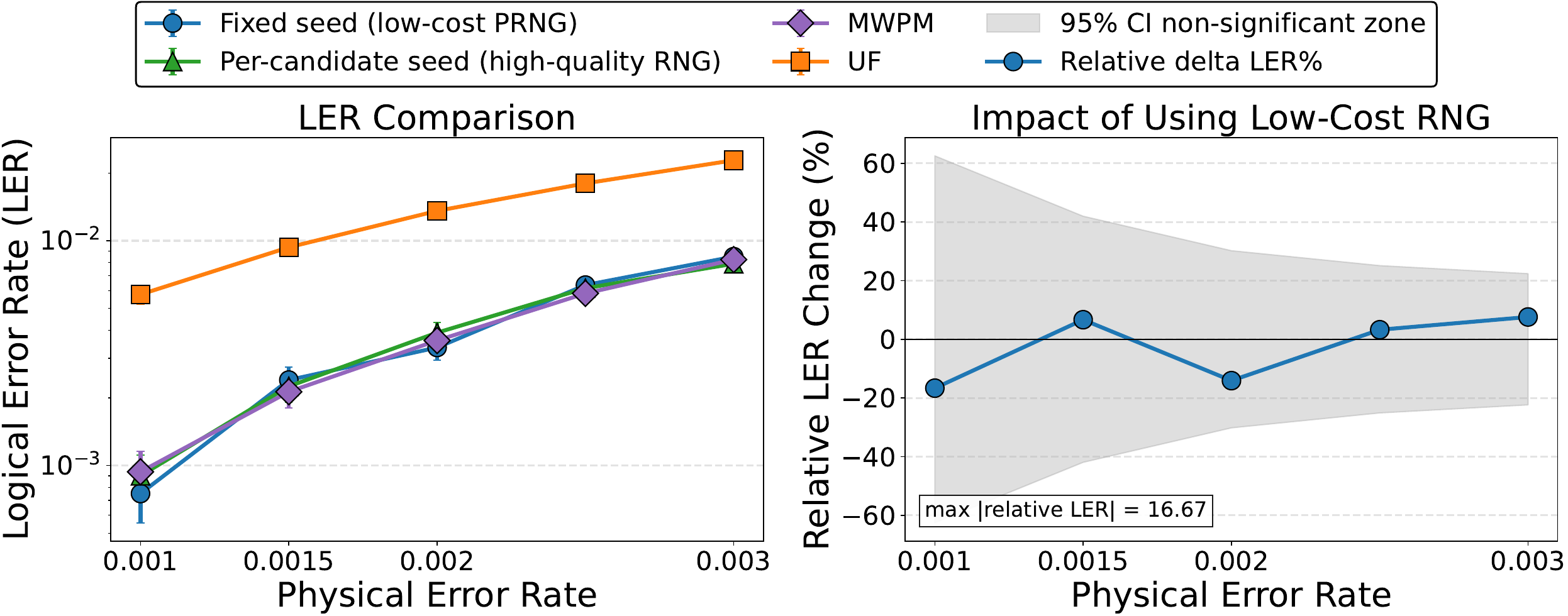}
  \caption{Decoding performance of different random policies.
  }
  \label{fig:fixed_rng_reference}
\end{figure}

\subsection{Robustness to Low-Cost Randomness}

A practical concern for FPGA deployment is whether decoder gains rely on high-quality, fully independent random streams, which are often expensive to implement in hardware.
To evaluate this risk in a conservative setting, we run all experiments with a \emph{fixed-seed} random policy during decoding.
In our implementation, each decoding shot uses a single stateful PRNG stream initialized from a fixed base seed.

We evaluate our decoder under different random policies.
\figref{fig:fixed_rng_reference} shows the LER comparison and quantifies the relative LER difference between different PRNGs.
The shaded region indicates a 95\% non-significant zone using a binomial approximation. 
Our work remains competitive against both MWPM and UF when using cheap but low-quality PRNGs.
The relative differences are small and mostly lie within the 95\% non-significant band, indicating no meaningful instability from the low-cost fixed-seed setting.

\begin{figure}[!tbp]
    \centering
    \includegraphics[width=0.95\linewidth]{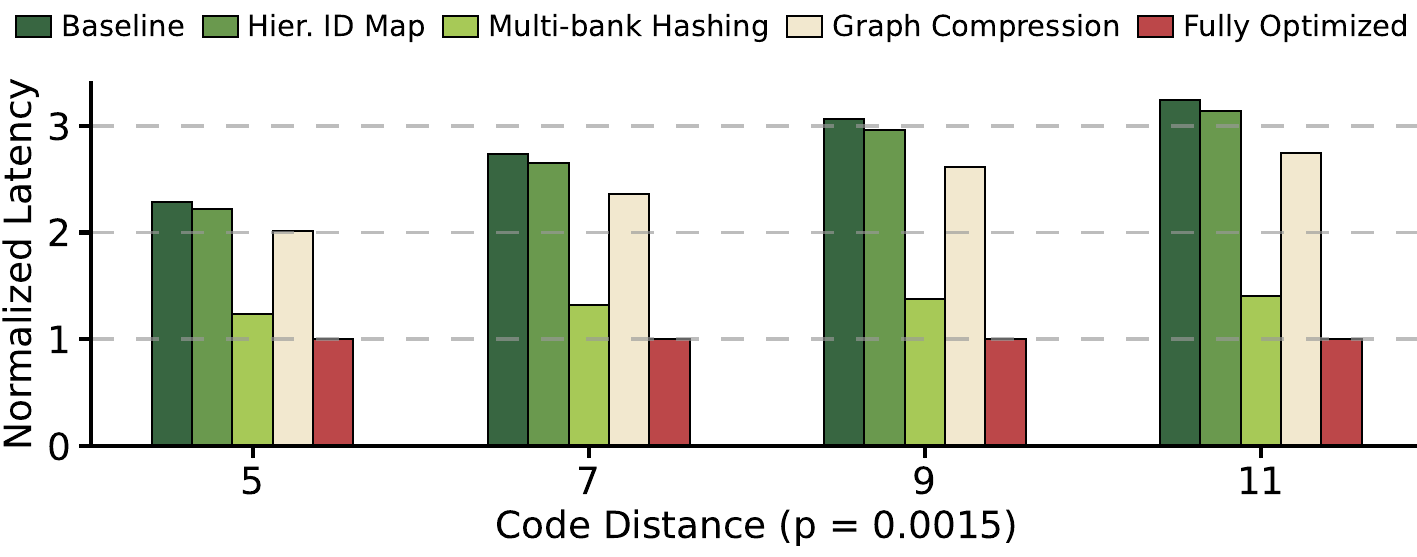}
    \caption{Decoding latency with and without the proposed optimizations.}
    \label{fig:ablation:latency}
\end{figure} 

\subsection{Optimization Ablation}

We perform an ablation study to quantify the impact of different optimizations. We use a hardware-oriented software simulator that mirrors the dataflow of our final microarchitecture and allows individual hardware features to be selectively enabled or disabled. 

We take the coset ensemble decoder architecture without any additional optimization as our baseline. 
Through ablation experiments, we aim to demonstrate that these optimizations can work synergistically to achieve an overall reduction in decoding latency. 
The trend in \figref{fig:ablation:latency} shows that the benefits of our optimizations increase with larger code distance \(d\). At \(d=11\) and \(p=0.0015\), Hierarchical ID Mapping delivers \(1.03\times\), Multi-bank Hashing delivers \(2.30\times\), and Graph Compression delivers \(1.18\times\) speedup over the baseline; enabling all optimizations achieves a \(3.24\times\) overall speedup.

\section{Related Work}\label{sec:relatedworks}

\subsection{QEC Algorithm}
While there exist families of quantum error correction codes and decoding algorithms~\cite{tillich2013quantum, leverrier2015quantum, fawzi2020constant, fawzi2018efficient, panteleev2021quantum, panteleev2021degenerate}, our work focuses on the surface code~\cite{kitaev2003fault,deMartiiOlius2024}. Two mainstream decoding methods are Minimum-Weight Perfect Matching (MWPM)~\cite{Higgott2025} and Union-Find (UF)~\cite{delfosse2021almost}. MWPM solves the physical ML error problem while UF is a faster, sub-optimal version. Our decoder, by accounting for degeneracy and logical cosets, achieves higher accuracy than UF-based decoders in the LER comparisons (\figref{fig:LER} and \figref{fig:LER_larger_distance}) while remaining in a similar low-latency regime (\figref{fig:latency}); compared with MWPM-based Micro-Blossom, it provides comparable accuracy at significantly lower latency. In contrast, the Tensor-Network (TN) decoder~\cite{deMartiiOlius2024}, which directly solves the logical coset ML problem, suffers from high contraction complexity. Although our decoder solves a sub-optimal coset ML problem, it maintains very low latency and high scalability for real-time implementation.

\subsection{QEC Hardware and Compilation}
There have been various hardware implementations of MWPM and UF decoders and their variants~\cite{das2022lilliput, vittal2023astrea, wu2023fusionblossomfastmwpm, wu2025micro, alavisamani2024promatch, liyanage2023scalable, Liyanage2024, valentino2025quekuf, liang2026hardware}. 
In particular, ~\cite{Liyanage2024,valentino2025quekuf} optimize the Union-Find clustering phase by mapping many vertices or clusters to distributed processing elements and exploiting extensive spatial parallelism, achieving low latency at the cost of considerable hardware resources.
By contrast, our design uses a single deeply pipelined clustering data path and explicitly reduces pipeline stalls via forwarding and conflict-aware memory organization, attaining low latency while remaining significantly more resource-efficient.
In parallel, recent efforts have sought to automate and accelerate the generation of detector error models for QEC protocols involving logical operations~\cite{fang2026lightstim,ziad2026greenpeas}.

\section{Conclusion}
This work presents a novel algorithm--hardware co-design for quantum error correction (QEC) decoding.
At the algorithmic level, we propose coset ensemble decoding, together with reverse-order elimination and lossless graph compression, to approximate coset-level maximum-likelihood decoding at practical cost.
At the hardware level, we introduce a customized architecture with multi-bank memory hashing and hierarchical ID mapping, achieving high hardware efficiency and scalability.
Overall, our co-design achieves a better accuracy--latency trade-off than prior state-of-the-art decoders, while maintaining high resource efficiency for scalable QEC decoding.
Future work includes assessing the practicality of this approach on available quantum devices, evaluating its performance across broader code families and noise settings, and automating the proposed optimizations across a broader range of FPGA implementations.

\section*{Acknowledgment}
The support of the UK EPSRC (Grant EP/W03221X/1, EP/V028251/1, EP/S030069/1, EP/X036006/1), UKRI (Grant 256), Altera and AMD is gratefully acknowledged.

\clearpage

\bibliographystyle{IEEEtran}
\bibliography{refs}

@article{gidney2021stim,
  title={Stim: a fast stabilizer circuit simulator},
  author={Gidney, Craig},
  journal={Quantum},
  volume={5},
  pages={497},
  year={2021},
  publisher={Verein zur F{\"o}rderung des Open Access Publizierens in den Quantenwissenschaften}
}

@article{bausch2023learning,
  title={Learning to decode the surface code with a recurrent, transformer-based neural network},
  author={Bausch, Johannes and Senior, Andrew W and Heras, Francisco J H and Edlich, Thomas and Davies, Alex and Newman, Michael and Jones, Cody and Satzinger, Kevin and Niu, Murphy Yuezhen and Blackwell, Sam and Holland, George and Kafri, Dvir and Atalaya, Juan and Gidney, Craig and Hassabis, Demis and Boixo, Sergio and Neven, Hartmut and Kohli, Pushmeet},
  journal={arXiv preprint arXiv:2310.05900},
  year={2023}
}

@article{Kolmogorov2009BlossomVA,
  title={Blossom V: a new implementation of a minimum cost perfect matching algorithm},
  author={Vladimir Kolmogorov},
  journal={Mathematical Programming Computation},
  year={2009},
  volume={1},
  pages={43-67},
  url={https://api.semanticscholar.org/CorpusID:17864814}
}

@article{wang2009threshold,
  title={Threshold error rates for the toric and surface codes},
  author={Wang, David S and Fowler, Austin G and Stephens, Ashley M and Hollenberg, Lloyd CL},
  journal={arXiv preprint arXiv:0905.0531},
  year={2009}
}

@inproceedings{das2022afs,
  title={Afs: Accurate, fast, and scalable error-decoding for fault-tolerant quantum computers},
  author={Das, Poulami and Pattison, Christopher A and Manne, Srilatha and Carmean, Douglas M and Svore, Krysta M and Qureshi, Moinuddin and Delfosse, Nicolas},
  booktitle={2022 IEEE International Symposium on High-Performance Computer Architecture (HPCA)},
  pages={259--273},
  year={2022},
  organization={IEEE}
}

@inproceedings{holmes2020nisq,
  title={NISQ+: Boosting quantum computing power by approximating quantum error correction},
  author={Holmes, Adam and Jokar, Mohammad Reza and Pasandi, Ghasem and Ding, Yongshan and Pedram, Massoud and Chong, Frederic T},
  booktitle={2020 ACM/IEEE 47th annual international symposium on computer architecture (ISCA)},
  pages={556--569},
  year={2020},
  organization={IEEE}
}

@inproceedings{wu2025micro,
  title={Micro blossom: Accelerated minimum-weight perfect matching decoding for quantum error correction},
  author={Wu, Yue and Liyanage, Namitha and Zhong, Lin},
  booktitle={Proceedings of the 30th ACM International Conference on Architectural Support for Programming Languages and Operating Systems, Volume 2},
  pages={639--654},
  year={2025}
}

@article{kjaergaard2020superconducting,
  title={Superconducting Qubits: Current State of Play},
  author={Kjaergaard, Morten and Schwartz, Maximilian E and Braum{\"u}ller, Jochen and Krantz, Philip and Wang, Joel I-J and Gustavsson, Simon and Oliver, William D},
  journal={Annual Review of Condensed Matter Physics},
  volume={11},
  pages={369--395},
  year={2020},
  publisher={Annual Reviews}
}

@article{terhal2015quantum,
  title={Quantum error correction for quantum memories},
  author={Terhal, Barbara M},
  journal={Reviews of Modern Physics},
  volume={87},
  number={2},
  pages={307},
  year={2015},
  publisher={APS}
}

@article{aharonov1997ftthreshold,
  title={Fault-tolerant quantum computation with constant error rate},
  author={Aharonov, Dorit and Ben-Or, Michael},
  journal={arXiv preprint quant-ph/9611025},
  year={1997}
}

@article{knill1998resilient,
  title={Resilient quantum computation: error models and thresholds},
  author={Knill, Emanuel and Laflamme, Raymond and Zurek, Wojciech H},
  journal={Proceedings of the Royal Society of London. Series A: Mathematical, Physical and Engineering Sciences},
  volume={454},
  number={1969},
  pages={365--384},
  year={1998},
  publisher={The Royal Society}
}

@article{fowler2012surface,
  title={Surface codes: Towards practical large-scale quantum computation},
  author={Fowler, Austin G and Mariantoni, Matteo and Martinis, John M and Cleland, Andrew N},
  journal={Physical Review A},
  volume={86},
  number={3},
  pages={032324},
  year={2012},
  publisher={APS}
}

@article{dennis2002topological,
  title={Topological quantum memory},
  author={Dennis, Eric and Kitaev, Alexei and Landahl, Andrew and Preskill, John},
  journal={Journal of Mathematical Physics},
  volume={43},
  number={9},
  pages={4452--4505},
  year={2002},
  publisher={AIP Publishing}
}

@article{delfosse2021almost,
  title={Almost-linear time decoding algorithm for topological codes},
  author={Delfosse, Nicolas and Nickerson, Naomi H.},
  journal={Quantum},
  volume={5},
  pages={595},
  year={2021},
  publisher={Verein zur F{\"o}rderung des Open Access Publizierens in den Quantenwissenschaften}
}

@article{Roffe2019,
   title={Quantum error correction: an introductory guide},
   volume={60},
   ISSN={1366-5812},
   url={http://dx.doi.org/10.1080/00107514.2019.1667078},
   DOI={10.1080/00107514.2019.1667078},
   number={3},
   journal={Contemporary Physics},
   publisher={Informa UK Limited},
   author={Roffe, Joschka},
   year={2019},
   month=jul, pages={226--245} }

@article{surfacecodebelowthreshold2024,
   title={Quantum error correction below the surface code threshold},
   volume={638},
   ISSN={1476-4687},
   url={http://dx.doi.org/10.1038/s41586-024-08449-y},
   DOI={10.1038/s41586-024-08449-y},
   number={8052},
   journal={Nature},
   publisher={Springer Science and Business Media LLC},
   author={Acharya, Rajeev and Abanin, Dmitry A. and Aghababaie-Beni, Laleh and Aleiner, Igor and Andersen, Trond I. and Ansmann, Markus and Arute, Frank and Arya, Kunal and Asfaw, Abraham and Astrakhantsev, Nikita and Atalaya, Juan and Babbush, Ryan and Bacon, Dave and Ballard, Brian and Bardin, Joseph C. and Bausch, Johannes and Bengtsson, Andreas and Bilmes, Alexander and Blackwell, Sam and Boixo, Sergio and Bortoli, Gina and Bourassa, Alexandre and Bovaird, Jenna and Brill, Leon and Broughton, Michael and Browne, David A. and Buchea, Brett and Buckley, Bob B. and Buell, David A. and Burger, Tim and Burkett, Brian and Bushnell, Nicholas and Cabrera, Anthony and Campero, Juan and Chang, Hung-Shen and Chen, Yu and Chen, Zijun and Chiaro, Ben and Chik, Desmond and Chou, Charina and Claes, Jahan and Cleland, Agnetta Y. and Cogan, Josh and Collins, Roberto and Conner, Paul and Courtney, William and Crook, Alexander L. and Curtin, Ben and Das, Sayan and Davies, Alex and De Lorenzo, Laura and Debroy, Dripto M. and Demura, Sean and Devoret, Michel and Di Paolo, Agustin and Donohoe, Paul and Drozdov, Ilya and Dunsworth, Andrew and Earle, Clint and Edlich, Thomas and Eickbusch, Alec and Elbag, Aviv Moshe and Elzouka, Mahmoud and Erickson, Catherine and Faoro, Lara and Farhi, Edward and Ferreira, Vinicius S. and Burgos, Leslie Flores and Forati, Ebrahim and Fowler, Austin G. and Foxen, Brooks and Ganjam, Suhas and Garcia, Gonzalo and Gasca, Robert and Genois, {\'E}lie and Giang, William and Gidney, Craig and Gilboa, Dar and Gosula, Raja and Dau, Alejandro Grajales and Graumann, Dietrich and Greene, Alex and Gross, Jonathan A. and Habegger, Steve and Hall, John and Hamilton, Michael C. and Hansen, Monica and Harrigan, Matthew P. and Harrington, Sean D. and Heras, Francisco J. H. and Heslin, Stephen and Heu, Paula and Higgott, Oscar and Hill, Gordon and Hilton, Jeremy and Holland, George and Hong, Sabrina and Huang, Hsin-Yuan and Huff, Ashley and Huggins, William J. and Ioffe, Lev B. and Isakov, Sergei V. and Iveland, Justin and Jeffrey, Evan and Jiang, Zhang and Jones, Cody and Jordan, Stephen and Joshi, Chaitali and Juhas, Pavol and Kafri, Dvir and Kang, Hui and Karamlou, Amir H. and Kechedzhi, Kostyantyn and Kelly, Julian and Khaire, Trupti and Khattar, Tanuj and Khezri, Mostafa and Kim, Seon and Klimov, Paul V. and Klots, Andrey R. and Kobrin, Bryce and Kohli, Pushmeet and Korotkov, Alexander N. and Kostritsa, Fedor and Kothari, Robin and Kozlovskii, Borislav and Kreikebaum, John Mark and Kurilovich, Vladislav D. and Lacroix, Nathan and Landhuis, David and Lange-Dei, Tiano and Langley, Brandon W. and Laptev, Pavel and Lau, Kim-Ming and Le Guevel, Lo{\"i}ck and Ledford, Justin and Lee, Joonho and Lee, Kenny and Lensky, Yuri D. and Leon, Shannon and Lester, Brian J. and Li, Wing Yan and Li, Yin and Lill, Alexander T. and Liu, Wayne and Livingston, William P. and Locharla, Aditya and Lucero, Erik and Lundahl, Daniel and Lunt, Aaron and Madhuk, Sid and Malone, Fionn D. and Maloney, Ashley and Mandr{\`a}, Salvatore and Manyika, James and Martin, Leigh S. and Martin, Orion and Martin, Steven and Maxfield, Cameron and McClean, Jarrod R. and McEwen, Matt and Meeks, Seneca and Megrant, Anthony and Mi, Xiao and Miao, Kevin C. and Mieszala, Amanda and Molavi, Reza and Molina, Sebastian and Montazeri, Shirin and Morvan, Alexis and Movassagh, Ramis and Mruczkiewicz, Wojciech and Naaman, Ofer and Neeley, Matthew and Neill, Charles and Nersisyan, Ani and Neven, Hartmut and Newman, Michael and Ng, Jiun How and Nguyen, Anthony and Nguyen, Murray and Ni, Chia-Hung and Niu, Murphy Yuezhen and O'Brien, Thomas E. and Oliver, William D. and Opremcak, Alex and Ottosson, Kristoffer and Petukhov, Andre and Pizzuto, Alex and Platt, John and Potter, Rebecca and Pritchard, Orion and Pryadko, Leonid P. and Quintana, Chris and Ramachandran, Ganesh and Reagor, Matthew J. and Redding, John and Rhodes, David M. and Roberts, Gabrielle and Rosenberg, Eliott and Rosenfeld, Emma and Roushan, Pedram and Rubin, Nicholas C. and Saei, Negar and Sank, Daniel and Sankaragomathi, Kannan and Satzinger, Kevin J. and Schurkus, Henry F. and Schuster, Christopher and Senior, Andrew W. and Shearn, Michael J. and Shorter, Aaron and Shutty, Noah and Shvarts, Vladimir and Singh, Shraddha and Sivak, Volodymyr and Skruzny, Jindra and Small, Spencer and Smelyanskiy, Vadim and Smith, W. Clarke and Somma, Rolando D. and Springer, Sofia and Sterling, George and Strain, Doug and Suchard, Jordan and Szasz, Aaron and Sztein, Alex and Thor, Douglas and Torres, Alfredo and Torunbalci, M. Mert and Vaishnav, Abeer and Vargas, Justin and Vdovichev, Sergey and Vidal, Guifre and Villalonga, Benjamin and Heidweiller, Catherine Vollgraff and Waltman, Steven and Wang, Shannon X. and Ware, Brayden and Weber, Kate and Weidel, Travis and White, Theodore and Wong, Kristi and Woo, Bryan W. K. and Xing, Cheng and Yao, Z. Jamie and Yeh, Ping and Ying, Bicheng and Yoo, Juhwan and Yosri, Noureldin and Young, Grayson and Zalcman, Adam and Zhang, Yaxing and Zhu, Ningfeng and Zobrist, Nicholas},
   year={2024},
   month=dec, pages={920--926} }

@misc{gottesman1997stabilizercodesquantumerror,
      title={Stabilizer Codes and Quantum Error Correction}, 
      author={Daniel Gottesman},
      year={1997},
      eprint={quant-ph/9705052},
      archivePrefix={arXiv},
      primaryClass={quant-ph},
      url={https://arxiv.org/abs/quant-ph/9705052}, 
}

@misc{poulin2008iterativedecodingsparsequantum,
      title={On the iterative decoding of sparse quantum codes}, 
      author={David Poulin and Yeojin Chung},
      year={2008},
      eprint={0801.1241},
      archivePrefix={arXiv},
      primaryClass={quant-ph},
      url={https://arxiv.org/abs/0801.1241}, 
}

@article{Stace2010,
   title={Error correction and degeneracy in surface codes suffering loss},
   volume={81},
   ISSN={1094-1622},
   url={http://dx.doi.org/10.1103/PhysRevA.81.022317},
   DOI={10.1103/physreva.81.022317},
   number={2},
   journal={Physical Review A},
   publisher={American Physical Society (APS)},
   author={Stace, Thomas M. and Barrett, Sean D.},
   year={2010},
   month=feb }

@ARTICLE{9456887,
  author={Fuentes, Patricio and Etxezarreta Martinez, Josu and Crespo, Pedro M. and Garcia-Fr{\'i}as, Javier},
  journal={IEEE Access}, 
  title={Degeneracy and Its Impact on the Decoding of Sparse Quantum Codes}, 
  year={2021},
  volume={9},
  number={},
  pages={89093-89119},
  doi={10.1109/ACCESS.2021.3089829}}

@misc{wu2023fusionblossomfastmwpm,
      title={Fusion Blossom: Fast MWPM Decoders for QEC}, 
      author={Yue Wu and Lin Zhong},
      year={2023},
      eprint={2305.08307},
      archivePrefix={arXiv},
      primaryClass={quant-ph},
      url={https://arxiv.org/abs/2305.08307}, 
}

@article{Liyanage2024,
   title={{FPGA-Based Distributed Union-Find Decoder for Surface Codes}},
   volume={5},
   ISSN={2689-1808},
   url={http://dx.doi.org/10.1109/TQE.2024.3467271},
   DOI={10.1109/tqe.2024.3467271},
   journal={IEEE Transactions on Quantum Engineering},
   publisher={Institute of Electrical and Electronics Engineers (IEEE)},
   author={Liyanage, Namitha and Wu, Yue and Tagare, Siona and Zhong, Lin},
   year={2024},
   pages={1--18} }

@article{Higgott2025,
   title={Sparse Blossom: correcting a million errors per core second with minimum-weight matching},
   author={Higgott, Oscar and Gidney, Craig},
   journal={Quantum},
   volume={9},
   ISSN={2521-327X},
   url={http://dx.doi.org/10.22331/q-2025-01-20-1600},
   DOI={10.22331/q-2025-01-20-1600},
   publisher={Verein zur Forderung des Open Access Publizierens in den Quantenwissenschaften},
   year={2025},
   month=jan, pages={1600} }

@inproceedings{liyanage2023scalable,
  title={{Scalable quantum error correction for surface codes using FPGA}},
  author={Liyanage, Namitha and Wu, Yue and Deters, Alexander and Zhong, Lin},
  booktitle={2023 IEEE International Conference on Quantum Computing and Engineering (QCE)},
  volume={1},
  pages={916--927},
  year={2023},
  organization={IEEE}
}

@inproceedings{liang2026hardware,
  title={Hardware-Efficient Union-Find Decoder Towards Scalable Topological Quantum Codes},
  author={Liang, Shuang and Xu, Jubo and Lu, Yuncheng and Chen, Hao Mark and Yuan, Bo and Fan, Hongxiang},
  booktitle={2026 31st Asia and South Pacific Design Automation Conference (ASP-DAC)},
  pages={77--82},
  year={2026},
  organization={IEEE}
}

@article{ziad2026greenpeas,
  title={GreenPeas: Unlocking Adaptive Quantum Error Correction with Just-in-Time Decoding Hypergraphs},
  author={Ziad, Abbas B and Xu, Jubo and Fan, Hongxiang},
  journal={arXiv preprint arXiv:2604.16613},
  year={2026}
}

@article{fang2026lightstim,
  title={LightStim: A Framework for QEC Protocol Evaluation and Prototyping with Automated DEM Construction},
  author={Fang, Xiang and Wang, Ming and Wu, Yue and Prabhu, Sharanya and Tullsen, Dean and Miniskar, Narasinga Rao and Mueller, Frank and Humble, Travis and Ding, Yufei},
  journal={arXiv preprint arXiv:2604.21472},
  year={2026}
}

@inproceedings{vittal2023astrea,
  title={Astrea: Accurate quantum error-decoding via practical minimum-weight perfect-matching},
  author={Vittal, Suhas and Das, Poulami and Qureshi, Moinuddin},
  booktitle={Proceedings of the 50th Annual International Symposium on Computer Architecture},
  pages={1--16},
  year={2023}
}

@inproceedings{das2022lilliput,
  title={Lilliput: a lightweight low-latency lookup-table decoder for near-term quantum error correction},
  author={Das, Poulami and Locharla, Aditya and Jones, Cody},
  booktitle={Proceedings of the 27th ACM International Conference on Architectural Support for Programming Languages and Operating Systems},
  pages={541--553},
  year={2022}
}

@inproceedings{alavisamani2024promatch,
  title={Promatch: Extending the Reach of Real-Time Quantum Error Correction with Adaptive Predecoding},
  author={Alavisamani, Narges and Vittal, Suhas and Ayanzadeh, Ramin and Das, Poulami and Qureshi, Moinuddin},
  booktitle={Proceedings of the 29th ACM International Conference on Architectural Support for Programming Languages and Operating Systems, Volume 3},
  pages={818--833},
  year={2024}
}

@article{valentino2025quekuf,
  title={{QUEKUF: an FPGA Union Find Decoder for Quantum Error Correction on the Toric Code}},
  author={Valentino, Federico and Branchini, Beatrice and Conficconi, Davide and Sciuto, Donatella and Santambrogio, Marco D},
  journal={ACM Transactions on Reconfigurable Technology and Systems},
  year={2025},
  publisher={ACM New York, NY}
}

@article{tillich2013quantum,
  title={Quantum LDPC codes with positive rate and minimum distance proportional to the square root of the blocklength},
  author={Tillich, Jean-Pierre and Z{\'e}mor, Gilles},
  journal={IEEE Transactions on Information Theory},
  volume={60},
  number={2},
  pages={1193--1202},
  year={2014},
  publisher={IEEE}
}

@inproceedings{leverrier2015quantum,
  title={Quantum expander codes},
  author={Leverrier, Anthony and Tillich, Jean-Pierre and Z{\'e}mor, Gilles},
  booktitle={2015 IEEE 56th Annual Symposium on Foundations of Computer Science},
  pages={810--824},
  year={2015},
  organization={IEEE}
}

@article{fawzi2020constant,
  title={Constant overhead quantum fault tolerance with quantum expander codes},
  author={Fawzi, Omar and Grospellier, Antoine and Leverrier, Anthony},
  journal={Communications of the ACM},
  volume={64},
  number={1},
  pages={106--114},
  year={2021},
  publisher={ACM New York, NY, USA}
}

@inproceedings{fawzi2018efficient,
  title={Efficient decoding of random errors for quantum expander codes},
  author={Fawzi, Omar and Grospellier, Antoine and Leverrier, Anthony},
  booktitle={Proceedings of the 50th Annual ACM SIGACT Symposium on Theory of Computing},
  pages={521--534},
  year={2018}
}

@article{panteleev2021quantum,
  title={Quantum LDPC codes with almost linear minimum distance},
  author={Panteleev, Pavel and Kalachev, Gleb},
  journal={IEEE Transactions on Information Theory},
  volume={68},
  number={1},
  pages={213--229},
  year={2022},
  publisher={IEEE}
}

@article{panteleev2021degenerate,
  title={Degenerate quantum LDPC codes with good finite length performance},
  author={Panteleev, Pavel and Kalachev, Gleb},
  journal={Quantum},
  volume={5},
  pages={585},
  year={2021},
  publisher={Verein zur F{\"o}rderung des Open Access Publizierens in den Quantenwissenschaften}
}

@article{kitaev2003fault,
  title={Fault-tolerant quantum computation by anyons},
  author={Kitaev, A Yu},
  journal={Annals of physics},
  volume={303},
  number={1},
  pages={2--30},
  year={2003},
  publisher={Elsevier}
}

@misc{higgott2021pymatchingpythonpackagedecoding,
      title={PyMatching: A Python package for decoding quantum codes with minimum-weight perfect matching}, 
      author={Oscar Higgott},
      year={2021},
      eprint={2105.13082},
      archivePrefix={arXiv},
      primaryClass={quant-ph},
      url={https://arxiv.org/abs/2105.13082}, 
}

@article{Hsieh2011,
   title={NP-hardness of decoding quantum error-correction codes},
   volume={83},
   ISSN={1094-1622},
   url={http://dx.doi.org/10.1103/PhysRevA.83.052331},
   DOI={10.1103/physreva.83.052331},
   number={5},
   journal={Physical Review A},
   publisher={American Physical Society (APS)},
   author={Hsieh, Min-Hsiu and Le Gall, Fran{\c c}ois},
   year={2011},
   month=may }

@article{deMartiiOlius2024,
   title={Decoding algorithms for surface codes},
   volume={8},
   ISSN={2521-327X},
   url={http://dx.doi.org/10.22331/q-2024-10-10-1498},
   DOI={10.22331/q-2024-10-10-1498},
   journal={Quantum},
   publisher={Verein zur Forderung des Open Access Publizierens in den Quantenwissenschaften},
   author={deMarti iOlius, Antonio and Fuentes, Patricio and Or{\'u}s, Rom{\'a}n and Crespo, Pedro M. and Etxezarreta Martinez, Josu},
   year={2024},
   month=oct, pages={1498} }

\end{document}